\documentclass[sigconf]{acmart}

\usepackage{color}
\usepackage{amsmath}
\usepackage{amsfonts}
\usepackage[colorinlistoftodos,prependcaption,textsize=footnotesize]{todonotes} 

\usepackage{thmtools}
\usepackage{thm-restate}

\usepackage[ruled,algosection,noend,linesnumbered]{algorithm2e}

\usepackage{amssymb}

\SetCommentSty{mycommfont}
\usepackage{pifont} 
\usepackage{comment}
\usepackage{amssymb}
\usepackage{soul,color}

\usepackage{booktabs}
\usepackage{makecell}
\usepackage{soul}
\usepackage{array}
\usepackage{xspace}
\usepackage{subfig}
\usepackage[colorinlistoftodos,prependcaption,textsize=footnotesize]{todonotes} 
\usepackage{pifont} 
\usepackage{epstopdf}

\newcommand{\name}{{Myst}\xspace}

\newcommand{\A}{\mathcal{A}}
\newcommand{\F}{\mathcal{F}}
\newcommand{\GG}{\mathcal{G}}
\newcommand{\Gr}{\mathbb{G}}
\newcommand{\Z}{\mathbb{Z}}

\usepackage{booktabs} 

\begin{document}

\title{A Touch of Evil: High-Assurance Cryptographic\\Hardware from Untrusted Components}

\author{Vasilios Mavroudis}
\affiliation{\institution{University College London}}
\email{v.mavroudis@cs.ucl.ac.uk}

\author{Andrea Cerulli}
\affiliation{\institution{University College London}}
\email{andrea.cerulli.13@ucl.ac.uk}

\author{Petr Svenda}
\affiliation{\institution{Masaryk University}}
\email{svenda@fi.muni.cz}

\author{Dan Cvrcek}
\affiliation{\institution{EnigmaBridge}}
\email{dan@enigmabridge.com}

\author{Dusan Klinec}
\affiliation{\institution{EnigmaBridge}}
\email{dusan@enigmabridge.com}

\author{George Danezis}
\affiliation{\institution{University College London}}
\email{g.danezis@ucl.ac.uk}

\renewcommand{\shortauthors}{V. Mavroudis et al.}

\begin{abstract}
The semiconductor industry is fully globalized and integrated circuits (ICs) are commonly defined, designed and fabricated in different premises across the world. This reduces production costs, but also exposes ICs to supply chain attacks, where insiders introduce malicious circuitry into the final products.  Additionally, despite 
extensive post-fabrication testing, it is not uncommon for ICs with subtle fabrication errors to make it into 
production systems. While many systems may be able to tolerate a few byzantine components,
this is not the case for cryptographic hardware, storing and computing on confidential data.
For this reason, many error and backdoor detection techniques have been proposed over the years.
So far all attempts have been either quickly circumvented, or come 
with unrealistically high manufacturing costs and complexity.

This paper proposes~\name, a practical high-assurance architecture, that uses commercial off-the-shelf (COTS) hardware, and provides strong security guarantees, even in the presence of multiple malicious or faulty components. The key idea 
is to combine protective-redundancy with modern threshold cryptographic techniques to build a system tolerant to hardware trojans and errors.
To evaluate our design, we build a Hardware Security Module that provides the highest 
level of assurance possible with COTS components.
Specifically, we employ more than a hundred COTS secure cryptocoprocessors, verified to FIPS140-2 Level 4 tamper-resistance standards, and use them to realize high-confidentiality random number generation, key derivation, public key decryption and signing. 
Our experiments show a reasonable computational overhead (less than ${1\%}$ for both Decryption and Signing) and an exponential increase in backdoor-tolerance as more ICs are added.
\end{abstract}

\keywords{cryptographic hardware; hardware trojans; backdoor-tolerance; secure architecture}





\maketitle


\section{Introduction}\label{sec:intro}
Many critical systems with high security needs rely on secure cryptoprocessors to 
carry out sensitive security tasks (e.g., key generation and storage, legally binding digital signature, code signing)
and provide a protection layer against cyber-attacks and security breaches.
These systems are typically servers handling sensitive data, 
banking infrastructure, military equipment and space stations.  
In most cases, secure cryptoprocessors come embedded into Hardware Security Modules, Trusted Platform Modules and Cryptographic Accelerators, which are assumed to be both secure and reliable.
This entails that errors in any of the Integrated Circuits (ICs) would be devastating for the security 
of the final system. For this reason, the design and fabrication of the underlying ICs must abide f
to high-quality specifications and standards. These ensure that there are no intentional 
or unintentional errors in the circuits, but more importantly ensure the integrity of the hardware supply chain. ~\cite{DBLP:conf/nsdi/KingTCGJZ08}.

Unfortunately, vendors are not always able to oversee all parts of the supply chain~\cite{force2005high, miyamoto2017counterfeit}. The constant reduction in transistor size makes IC fabrication an expensive process, 
and IC designers often outsource the fabrication task to overseas foundries to reduce their costs
~\cite{DBLP:journals/computer/EdenfeldKRZ04, heck2011creating, yeh2012trends}.
This limits vendors to run only post-fabrication tests to uncover potential defects. Those tests are
very efficient against common defects, but subtle errors are hard to uncover. For instance, 
cryptoprocessors with defective RNG modules and hardware cipher implementations have made it into production 
in the past~\cite{fredriksson2016case, courtois2009dark}.

Additionally, parts of the IC's supply chain are left vulnerable to attacks from malicious insiders~\cite{tehranipoor2011introduction, DBLP:conf/dac/PotkonjakNNM09, united2005defense, DBLP:journals/jce/BeckerRPB14}
and have a higher probability of introducing unintentional errors in the final product.
In several documented real-world cases, contained errors, backdoors or trojan horses. For instance, recently an exploitable vulnerability was discovered on Intel processors that utilize Intel Active Management Technology (AMT)~\cite{www-intel}, while vulnerable ICs have been reported in military~\cite{DBLP:conf/ches/SkorobogatovW12, mitra2015trojan} applications, networking equipment~\cite{jin2010hardware, gallagher2014photos}, and various other application~\cite{skorobogatov2012hardware, adee2008hunt, markoff2009old, shrimpton2015provable}.
Furthermore, the academic community has designed various types of hardware trojans (HT), and backdoors that demonstrate the extent of the problem and its mitigation difficulty~\cite{DBLP:conf/dft/WangMKNB12,wang2014hardware,wang2011sequential,DBLP:conf/cases/KutznerPS13,DBLP:journals/pieee/BhuniaHBN14, pellegrini2010fault, becker2013stealthy, borghoff2012prince}.

Due to the severity of these threats, there is a large body of work on the mitigation of malicious circuitry.
Existing works have pursued two different directions: detection and prevention.
Detection techniques aim to determine whether any HTs exist in a given circuit~\cite{DBLP:journals/tvlsi/WeiP14, 5778966, soll2014based, agrawal2007trojan},
while prevention techniques either impede the introduction of HTs, or enhance the efficiency of HT detection~\cite{waksman2011silencing, 7835561, DBLP:conf/dac/WangCHR16, rajendran2013split, waksman2010tamper}. 
Unfortunately, both detection and prevention techniques are brittle, as new threats
are able to circumvent them quickly~\cite{DBLP:conf/dac/WangCHR16}. For instance, analog malicious hardware~\cite{yang2016a2} was able to evade known defenses, including split manufacturing,
which is considered one of the most promising and effective prevention approaches. Furthermore, most prevention techniques come with a high manufacturing cost for higher levels of security~\cite{bhunia2014hardware, DBLP:conf/dac/WangCHR16, 7835561}, which contradicts the motives of fabrication outsourcing.
To make matters worse, vendors that use commercial off-the-shelf (COTS) components are 
constrained to use only post-fabrication detection techniques, which further limits their mitigation toolchest.
All in all, backdoors being triggered by covert means and mitigation countermeasures are in an arms race that seems favorable to the attackers~\cite{wei2013undetectable, DBLP:conf/dac/WangCHR16}.


In this paper, we propose \name a new approach to the problem of building trustworthy
cryptographic hardware. Instead of attempting to detect or prevent hardware trojans and errors,
we proposed and provide evidence to support the hypothesis:\\

\noindent\textit{``We can build high-assurance cryptographic hardware 
from a set of untrusted components,
as long as at least one of them is not compromised, even if 
benign and malicious components are indistinguishable.''}\\

Our key insight is that by combining established privacy enhancing technologies (PETs),
with mature fault-tolerant system architectures, we can distribute trust between 
multiple components originating from non-crossing supply chains, thus reducing the likelihood of compromises.
To achieve this, we deploy distributed cryptographic schemes on top of an
N-variant system architecture, and build a trusted platform
that supports a wide-range of commonly used cryptographic operations
(e.g., random number and key generation, decryption, signing).
This design draws from protective-redundancy and component diversification~\cite{chen1978n}
and is built on the assumption that multiple processing units
and communication controllers may be compromised by the same adversary.
However, unlike N-variant systems, instead of  replicating the computations on all processing units, \name uses multi-party cryptographic schemes to distribute 
the secrets so that the components hold only shares of the secrets (and not the secrets themselves), at all times. As long as one of the components remains honest, the secret cannot be
reconstructed or leaked. Moreover, we can tolerate 
two or more non-colluding adversaries who have compromised $100\%$ of the components.

Our proposed architecture is of particular interest for two distinct categories of hardware vendors:
\begin{itemize}
\item Design houses that outsource the fabrication of their ICs.
\item COTS vendors that rely on commercial components to build their high-assurance hardware.
\end{itemize}

Understandably, design houses have much better control over the IC fabrication and the supply chain, and
this allows them to take full advantage of our architecture. 
In particular, they can combine existing detection~\cite{DBLP:conf/ches/ChakrabortyWPPB09,rajendran2017logic, DBLP:conf/host/BaoXS15} and prevention techniques~\cite{7835561, DBLP:conf/dac/WangCHR16, rajendran2013split, chakraborty2009security} with our proposed design, 
to reduce the likelihood of compromises for individual components.
On the other hand, COTS vendors have less control as they have limited visibility in the fabrication process
and the supply chain. However, they can still mitigate risk by using ICs from sources,
known to run their own fabrication facilities.

To our knowledge, this is the first work that uses distributed cryptographic protocols to build
and evaluate a hardware module architecture that is tolerant to multiple components carrying trojans or errors.
The effectiveness of this class of protocols for the problem of hardware trojans has been also been studied in 
previous theoretical works~\cite{dziembowski2016private, cryptoeprint:2016:527}.

To summarize, this paper makes the following contributions:
\begin{itemize}
\itemsep-0.025em
\item \textbf{Concept:} We introduce \textit{backdoor-tolerance}, where a system 
can still preserve its security properties in the presence of multiple compromised components.

\item \textbf{Design:} We demonstrate how cryptographic schemes (\S\ref{sec:protocols}) can be combined with N-variant system architectures (\S\ref{sec:arch}), to build high-assurance systems. Moreover, we introduce a novel distributed signing scheme
based on the Schnorr blind signatures (\S\ref{sec:sign}).

\item \textbf{Implementation:} We implement the proposed architecture 
by building a custom board featuring 120 highly-tamper resistant ICs, and realize secure variants of random number and key generation, public key decryption and signing (\S\ref{sec:impl}).

\item \textbf{Optimizations:} We implement a number of optimizations to ensure the \name architecture is competitive in terms of performance compared to single ICs. Some optimizations also concern embedded mathematical libraries which are of independent interest.

\item \textbf{Evaluation:} We evaluate the performance of \name, and use micro-benchmarks to assess the cost of all operations and bottlenecks (\S\ref{sec:eval}). 
\end{itemize}

Related works and their relation to \name are discussed in Section~\ref{sec:related}.

\section{Preliminaries}\label{sec:sec_mod}
In this section, we introduce \textit{backdoor-tolerance},
and outline our security assumptions for adversaries residing
in the IC's supply chain.

\subsection{Definition}\label{sec:definitions}
A \textit{Backdoor-Tolerant} system is able to 
ensure the confidentiality of the secrets it stores, the
integrity of its computations and its availability, against an adversary with defined capabilities.
Such a system can prevent data leakages and protect the integrity of the 
generated keys and random numbers.
Note that the definition makes no distinction between 
honest design or fabrication mistakes and hardware trojans, and 
accounts only for the impact these errors have on the security properties of the system.





\subsection{Threat Model}\label{sec:threats}


We assume an adversary is able to incorporate errors (e.g., a hardware trojan) in some ICs but not all of them. This is because, ICs are manufactured in many different fabrication facilities and locations by different vendors and the adversary is assumed not to be able to breach all their supply chains. Malicious ICs aim to be as stealthy as possible, and 
conceal their presence as best as possible, while benign components may have hard to detect errors (e.g., intermittent faults) that cannot be trivially uncovered by post-fabrication tests. Hence, malicious and defective components are indistinguishable from benign/non-faulty ones.

The adversary can gain access to the secrets stored in the malicious ICs and may also 
breach the integrity of any cryptographic function run by the IC (e.g., a broken random 
number generator). Moreover, the adversary has full control over the communication buses used to transfer data to and from the ICs. Hence, they are able to exfiltrate any information on the bus and the channel controller, and inject and drop messages to launch attacks. Additionally, they are able to connect and issue commands to the ICs, and if a system features more that one malicious ICs, the adversary is able to orchestrate attacks using all of them (i.e., colluding ICs). We also assume that the adversary may use any other side-channel that a hardware Trojan has been programmed to emit on -- such as RF~\cite{appelbaum2013shopping}, acoustic~\cite{genkin2014rsa} or optical~\cite{backes2008compromising} channels. Additionally, the adversary can trigger malicious ICs to cease functioning (i.e., availability attacks). We also make standard cryptographic assumptions about the adversary being computationally bound.

Finally, we assume 
a software developer or device operator builds and signs the applications to be run on the ICs. From our point of view they are trusted to provide high-integrity software without 
backdoors or other flaws. This integrity concern may be tackled though standard high-
integrity software development techniques, such as security audits, public code 
repository trees, deterministic builds, etc. The operator is also trusted to initialize the device properly
and choose ICs and define diverse quorums so that the probability of compromises is minimized.


%
%
%
%
%
%

%

\section{Architecture}\label{sec:arch}
In this section, we introduce the \name architecture (build and evalutation in \S~\ref{sec:impl} and~\ref{sec:eval}), 
which draws from fault-tolerant designs and N-variant systems~\cite{strigini2005fault, chen1978n}. 
The proposed design is based on the thesis that given a diverse 
set of ${k}$ ICs sourced from ${k}$ untrusted suppliers, a trusted system can be build,
as long as at least one of them is honest.
Alternatively, our architecture can tolerate up to 100\% compromised or defective components,
if not all of them collude with each other.
As illustrated in Figure~\ref{fig:overview},
\name has three types of components: (1) a remote host, (2) an untrusted IC controller,
and (3) the processing ICs. 

\begin{figure}
  \centering

  \includegraphics[width=\columnwidth]{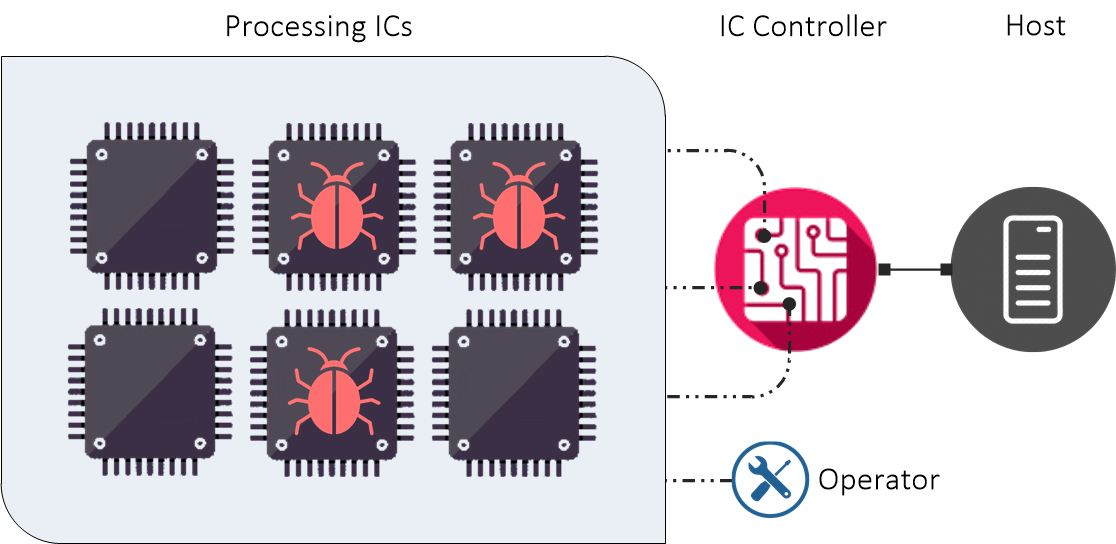}

  \caption{An overview of the \name 's distributed architecture, featuring all the integral components
  and communication buses. The gray area represents the cryptographic device, featuring several untrusted cryptoprocessors (ICs). As shown, the trusted operator interacts with individual ICs, while the host interacts with the device as a whole.}
  \label{fig:overview}
\end{figure}


\medskip \noindent{\bf{Processing ICs.}}
The ICs form the core of \name, as they are collectively entrusted 
to perform high-integrity computations, and provide secure storage space.
The ICs are programmable processors or microprocessors.
They have enough persistent memory to store keys and 
they feature a secure random number generator. Protection against physical 
tampering or side-channels is also desirable in most cases 
and mandated by the industry standards and best practices for 
cryptographic hardware.
For this reason, our prototype (Section~\ref{sec:impl}), 
comprises of components that verify to the 
a very high level of tamper-resistance (i.e., FIPS140-2 Level 4),
and feature a reliable random number generator (i.e., FIPS140-2 Level 3). 
Each implementation must feature two or more ICs, of which at least 
one must be free of backdoors. We define this coalition of ICs, as a \emph{quorum}.
The exact number of ICs in a quorum is determined by the 
user depending on the degree of assurance she wants to achieve (see also Subsection~\ref{sec:estimating} and the Assessment Subsection~\ref{sec:eval}). 

\medskip \noindent{\bf{IC controller.}}
The controller enables communication between the ICs, and makes
\name accessible to the outside world.
It manages the bus used by the processing ICs to communicate with each other
and the interface where user requests are delivered to.
Specifically, it enables: 
\begin{itemize}
\itemsep-0.025em
\item \textit{Unicast, IC-to-IC}: an IC A sends instructions an IC B, and receives the response.
\item \textit{Broadcast, IC-to-ICs}: an IC A broadcasts instructions to all other ICs, and receive their responses.
\item \textit{Unicast, Host-to-IC}: a remote client send commands to a specific IC, and receives the response.
\item \textit{Broadcast, Host-to-ICs}: a remote client broadcasts commands to all ICs, and receive their responses.
\end{itemize}

The controller is also an IC and is also untrusted. For instance it may drop, 
modify packets or forge new ones, in order to 
launch attacks against the protocols executed. It may also collude with 
one or more processing ICs.

\medskip \noindent{\bf{Operator.}}
The operator is responsible for setting up the system and configuring it. In particular, they are responsible for sourcing the ICs and making sure the quorums are as diverse as possible. Moreover, the operator determines the size of the quorums, 
and sets the security parameters of the cryptosystem (Section~\ref{sec:protocols}).
They are assumed to make a best effort to prevent compromises and may also be the owner of the secrets stored in the device.

\medskip \noindent{\bf{Remote Host.}}
The remote host connects to \name through the IC controller; it
issues high level commands and retrieves the results. The remote
host can be any kind of computer either in the local network or in a remote one.
In order for the commands to be executed by the ICs, the host must provide
proof of its authorization to issue commands, usually by signing them 
with a public key associated with the user's identity (e.g., certificate by a trusted CA).
Each command issued must include:
1) the requested operation, 2) any input data, and 3) the host's signature (see also Section~\ref{sec:acl}).
We note that a corrupt host may use \name to perform operations, but cannot extract
any secrets, or forge signatures (see also \S~\ref{sec:acl}).

\medskip \noindent{\bf{Communication Channels.}}
At the physical level, the ICs, the controller, and the hosts are connected through buses and network interfaces. Hence, all messages containing commands, as well as their eventual responses are routed to and from the ICs through untrusted buses. We assume that the adversary is able to eavesdrop on these buses and
tamper with their traffic (e.g., inject or modify packets). To address this, we use established cryptographic mechanisms to ensure integrity and confidentiality for transmitted data. More specifically, all unicast and broadcast packets
are signed using the sender IC's certificate, so that their origin and integrity 
can be verified by recipients. Moreover, in cases where the confidentiality of
the transmitted data needs to be protected, encrypted end-to-end channels are also established.
Such encrypted channels should be as lightweight as possible to minimize performance
overhead, yet able to guarantee both the confidentiality and the integrity of the transmitted data.

\subsection{Access Control Layer}\label{sec:acl}
Access Control (AC) is critical for all systems operating on confidential data.
In \name, AC determines which hosts can submit service requests to the system,
and which quorums they can interact with.
Despite the distributed architecture of \name, we can simply replicate  
established AC techniques on each IC.
More specifically, each IC is provided with the public keys of the hosts
that are allowed to have access to its shares and submit requests.
Optionally, this list may distinguish between hosts that have full access to the system,
and hosts that may only perform a subset of the operations.
Once a request is received, the IC verifies the signature of the host against this public key list. The list can be provided either when setting up \name, 
or when storing a secret or generating a key.

We note that it is up to the operator to decide
the parameters of each quorum (i.e., size ${k}$, ICs participating), and
provide the AC lists to the ICs. This is a critical procedure,
as if one of the hosts turns out to be compromised,
the quorum can be misused in order to either decrypt confidential
ciphertexts or sign bogus statements. However, the 
secrets stored in the quorum will remain safe as
there is no way for the adversary to extract them, unless they
use physical-tampering (which our prototype (\S\ref{sec:impl}) is also resilient against). This is also true in the case where one of the authorized hosts gets 
compromised. For instance, a malware taking over a host can use \name to sign
documents, but it cannot under any circumstances extract any secrets.

\subsection{Reliability Estimation}\label{sec:estimating}
In the case of cryptographic hardware, in order for the operator to decide on the threshold $k$ and the quorum composition, an estimation of the likelihood of hardware errors is needed. For this purpose we introduce \textit{k-tolerance}, which given $k$ foundries and an independent error probability, computes the probability that a quorum is secure.

\begin{equation}
\Pr[\mathsf{secure}] = 1-\Pr[\mathsf{error}]^k
\end{equation}\label{eq:estimate}

The quantification of the above parameters depends on the particular design and case, 
and as there is not commonly accepted way of evaluation, it largely depends 
on the information and testing capabilities each vendor has.
For instance, hardware designers that use split manufacturing~\cite{7835561, DBLP:conf/dac/WangCHR16, rajendran2013split} can estimate the probability of a backdoored component using the \textit{k-security} metric~\cite{imeson2013securing}.
On the other hand, COTS vendors cannot easily estimate the probability of a compromised component, as they are not 
always aware of the manufacturing technical details. Despite that, it is 
still possible for them to achieve very high levels of error and backdoor-tolerance
by increasing the size of the quorums and sourcing ICs from distinct facilities (i.e., minimizing the 
collusion likelihood).
It should be noted that as $k$ grows the cost increases linearly, while the 
security exponentially towards one.

\section{Secure distributed computations}\label{sec:protocols}
In this section, we introduce a set of protocols that leverage the diversity of
ICs in \name to realize standard cryptographic operations manipulating sensitive keying material. 
More specifically, our cryptosystem comprises of a key generation operation (\S\ref{sec:keygen}),
the ElGamal encryption operation (\S\ref{sec:enc}), 
distributed ElGamal decryption (\S\ref{sec:dec}), and
distributed signing based on Schnorr signatures (\S\ref{sec:sign})~\cite{elgamal1985public, fiat1986prove}.
These operations are realized using interactive cryptographic protocols 
that rely on standard cryptographic assumptions. For operational purposes, we also introduce
a key propagation protocol that enables secret sharing between non-overlapping quorums (\S\ref{sec:keyprop}).

Prior to the execution of any protocols, the ICs must be initialized with the 
domain parameters ${T=(p, a, b, G, n, h)}$ of the elliptic curve to be used, where
$p$ is a prime specifying the finite field ${\mathbb{F}_p}$, 
$a$ and ${b}$ are the curve coefficients, ${G}$ is the base point of the curve, with order ${n}$, 
and ${h}$ is the curve's cofactor. More details on the technical aspects of the initialization procedure are provided in the provisioning Section~\ref{sec:considerations}.

\medskip \noindent{\bf{Distributed Key Pair Generation.}}
The Distributed Key Pair Generation (DKPG) is a key building block
for most other protocols.
In a nutshell, DKPG enables a \textit{quorum} ${Q}$ of 
${k}$ ICs to collectively generate a 
random secret ${x}$, which is an element
of a finite field and a public value ${Y_{agg} = x \cdot G}$ 
for a given public elliptic curve point ${G}$. 
At the end of the DKPG protocol, each player holds
a share of the secret ${x}$, denoted as ${x_i}$
and the public value ${Y_{agg}}$. All steps of the protocol are illustrated in Figure~\ref{fig:DPG}, and explained in turn below.

\begin{figure}
  \centering
  \includegraphics[width=\columnwidth]{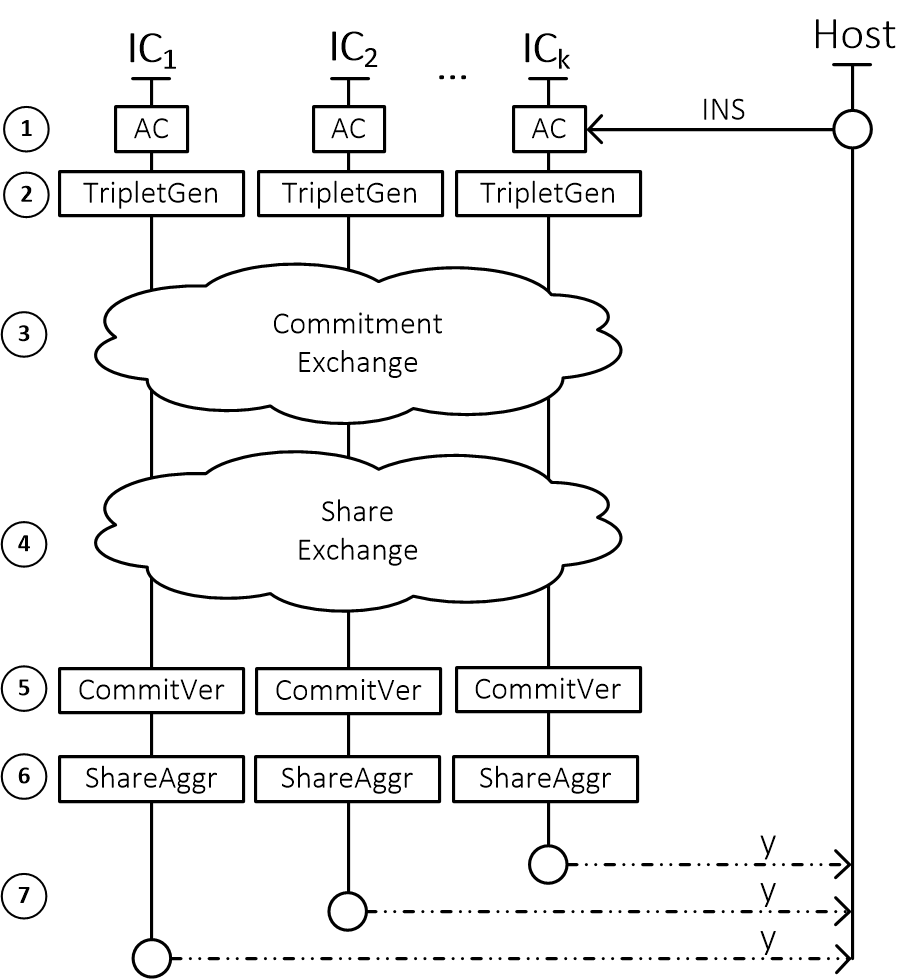}
  \caption{The interaction between the different participants during the execution of the Distributed
  Key Pair Generation (DKPG) protocol.}
  \label{fig:DPG}
\end{figure}

The execution of the protocol is triggered when an authorized host sends the corresponding instruction (\ding{182}).
At the first step of the protocol, each member of ${Q}$ runs Algorithm~\ref{alg:tripletgen}
and generates a triplet consisting of: 1) a share ${x_i}$, which is a randomly sampled element from
${\mathbb{Z}_n}$, 2) an elliptic curve point ${Y_i}$, and 3) a commitment to ${Y_i}$ denoted ${h_i}$. (\ding{183})

\begin{algorithm}
\SetAlgoLined
\SetKwInOut{Input}{Input}
    \SetKwInOut{Output}{Output}
    \Input{The domain parameters $\lambda$}
    \Output{A key triplet ${(x_i,Y_i,h_i)}$}
	$x_i \gets \text{Rand}(\lambda)$\\
	$Y_i \gets x_i \cdot G$\\
	$h_i \gets \text{Hash}(Y_i)$\\ 
	\Return ${(x_i,Y_i,h_i)}$\\
 \caption{TripletGen: Generation of a pair and its commitment.}
 \label{alg:tripletgen}
\end{algorithm}

Upon the generation of the triplet, the members perform a pairwise exchange of their commitments (\ding{184}), by the
end of which, they all hold a set ${H = \{h_1,h_2,..,h_t\}}$.
The commitment exchange terminates when 
${|H_q|=t\ \forall q \in Q}$. Another round of exchanges then starts,
this time for the shares of ${Y_{agg}}$ (\ding{185}) ${Y = \{Y_1,Y_2,..,Y_t\}}$.
The commitment exchange
round is of uttermost importance as it forces the participants
to commit to a share of ${Y_{agg}}$, before receiving 
the shares of others. This prevents attacks where an adversary first collects
the shares of others, and then crafts its share so as to bias the final pair, towards a secret key they know.

\begin{algorithm}
\SetAlgoLined
\SetKwInOut{Input}{Input}
    \SetKwInOut{Output}{Output}
    \Input{$(Y, H)$}
    \Output{Bool}
	\For{$i \in \{1,\ |Y|\}$}{
		\If{$\text{Hash} (Y_i)\neq h_i$}{
			\Return False\\	
		}
	\Return True\\
	}
\caption{CommitVerify: Checks ${Y_i}$, against their respective commitments.}
\label{alg:hashver}
\end{algorithm}

Once each member of the quorum receives $k$ shares (i.e., $|Y|=k$), it executes
Algorithm~\ref{alg:hashver} to verify the validity of
$Y$'s elements against their commitments in ${H}$. (\ding{186})
If one or more commitments fail the verification
then the member infers that an error (either intentional
or unintentional) occurred and the protocol is terminated.
Otherwise, if all commitments are successfully verified,
then the member executes Algorithm~\ref{alg:PointAgg} (\ding{187}) and
returns the result to the remote host (\ding{188}).
Note that it is important to return $Y_{agg}$, as well as the individual shares $Y_i$, as this 
protects against integrity attacks, where malicious ICs return a different share than the 
one they committed to during the protocol~\cite{DBLP:conf/crypto/Pedersen91, DBLP:journals/joc/GennaroJKR07}. Moreover, since $Y_i$ are shares of the public key, they are also assumed to be public, and available to any untrusted party.

\begin{algorithm}
\SetAlgoLined
\SetKwInOut{Input}{Input}
\SetKwInOut{Output}{Output}
\Input{Set of shares $Q$}
\Output{The aggregate of the shares ${q}$}
$q \gets0$\\
\For{$q_i \in Q$}{
	$q \gets q + q_i$\\
}
\Return ${q}$\\
 
 \caption{ShareAggr: Aggregates elements in a set of shares (e.g., ECPoints, field elements).}
 \label{alg:PointAgg}
\end{algorithm}


In the following sections, we rely on 
DKPG as a building block
of more complex operations.

\subsection{Distributed Public Key Generation}\label{sec:keygen}
The distributed key generation operation enables multiple ICs to collectively generate
a shared public and private key pair with shares distributed between all participating ICs. This is important in the presence of hardware trojans, as a single IC never gains access to the full private key at any point, while the integrity and secrecy of the pair is protected against maliciously crafted shares.

We opt for a scheme that offers the maximum level of confidentiality (\textit{t-of-t}, $k=t$), and
whose execution is identical to DKPG seen in Figure~\ref{fig:DPG}.
However, there are numerous protocols that allow for different thresholds, such as
Pedersen's VSS scheme~\cite{DBLP:conf/crypto/Pedersen91, DBLP:journals/joc/GennaroJKR07, DBLP:conf/acisp/StinsonS01}. 
The importance of the security threshold is discussed in more detail in Section ~\ref{sec:tradeoffs}.

Once a key pair has been generated,
the remote host can encrypt a plaintext using the public key $Y$,
request the decryption of a ciphertext, or
ask for a plaintext to be signed. In the
following sections we will outline the protocols
that realize these operations.


\subsection{Encryption}\label{sec:enc}

For encryption, we use the Elliptic Curve ElGamal scheme~\cite{elgamal1985public, DBLP:conf/icisc/Brandt05} (Algorithm~\ref{alg:Encryption}). This operation does not use the secret key, and can be performed directly on the host, or remotely by any party holding the public key, hence there is no need to perform it in a distributed manner.



\begin{algorithm}
\SetAlgoLined
\SetKwInOut{Input}{Input}
\SetKwInOut{Output}{Output}
\Input{The domain parameters $T$, the plaintext ${m}$ encoded as an element of the group $\mathbb{G}$, and the calculated public key ${Y_{agg}}$}
\Output{The Elgamal ciphertext tuple, ${(C_1,\ C_2)}$}
$r \gets Rand(T)$\\
$C_1 \gets r \cdot G$\\
$C_2 \gets m + r\cdot Y_{agg}$\\
\Return ${(C_1,\ C_2)}$\\ 
\caption{Encrypts a plaintext using the agreed common public key.}
\label{alg:Encryption}
\end{algorithm}

\subsection{Decryption}\label{sec:dec}
One of the fundamental cryptographic operations involving a private key is ciphertext decryption. 
In settings, where the key is held by a single authority, the 
decryption process is straightforward, 
but assumes that the hardware used to perform the decryption does not leak the secret decryption key. 
\name addresses this problem
by distributing the decryption process between ${k}$ distinct
processing ICs that hold shares of the key (Figure~\ref{fig:DEC}).

\begin{figure}
  \centering
  \includegraphics[width=\columnwidth]{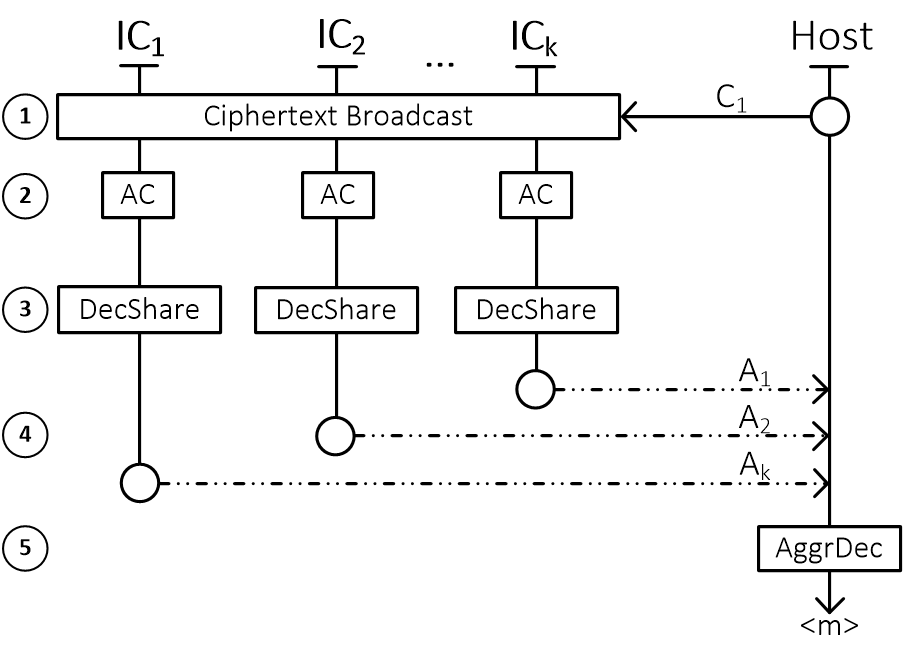}
  \caption{The interaction between the different ICs during the execution of the distributed
  decryption protocol.}
  \label{fig:DEC}
\end{figure}

The protocol runs as follows:
Initially, the host broadcasts the decryption instruction along with the part of the ciphertext $C_1$ to the processing ICs (\ding{182}). Upon reception, the ICs first verify
that the request is signed by an authorized user (\ding{183}), and then execute Algorithm~\ref{alg:DecShare}
to generate their decryption shares ${A_i}$ (\ding{184}). Once the shares are generated they are
sent back to the host, signed by the ICs and encrypted under the host's public key (\ding{185}). Once the host 
receives ${k}$ decryption shares, executes Algorithm~\ref{alg:DecCombinationFinal} to
combine them and retrieve the plaintext (\ding{186}).
 
\begin{algorithm}
\SetAlgoLined
\SetKwInOut{Input}{Input}
\SetKwInOut{Output}{Output}
\Input{A part of the Elgamal ciphertext ${(C_1)}$ and the IC's private key of ${x_i}$.}
\Output{The decryption share ${A_{i}}$, where ${i}$ is the IC's uid}
$A_{i} \gets -x_i \cdot C_1$\\
\Return ${A_{i}}$\\
\caption{DecShare: Returns the decryption share for a given ciphertext.}
\label{alg:DecShare}
\end{algorithm}

\begin{algorithm}
\SetAlgoLined
\SetKwInOut{Input}{Input}
\SetKwInOut{Output}{Output}
\Input{The Elgamal ciphertext ${C_2}$ and the set of decryption shares ${A}$.}
\Output{The plaintext ${m}$}

$D \gets 0$\\
\For{$A_i \in A$}{
	$D \gets D + A_i$\\
}
$m \gets (C_2 + D)$

\Return ${m}$\\
\caption{AggrDec:Combines the decryption shares and returns the plaintext for a given ciphertext.}
\label{alg:DecCombinationFinal}
\end{algorithm}

It should be noted that during the decryption process, the plaintext is not revealed to any other party except the host, and neither the secret key nor its shares ever leave the honest ICs. An extension to the above protocol can also prevent malicious ICs from returning arbitrary decryption shares, by incorporating a non-interactive zero knowledge proof~\cite{blum1988non} in the operation output.


\subsection{Random Number Generation}\label{sec:rng}
Another important application of secure hardware is the generation of a random fixed-length bit-strings
in the presence of active adversaries. The key property of such systems 
is that errors (e.g., a hardware backdoor), should not allow an adversary to 
bias the generated bitstring. 

The implementation of such an operation is straightforward. The 
remote host submits a request for randomness to all actors participating
in the quorum. Subsequently, each actor independently generates
a random share $b_i$, encrypts it with the public key of the host, and signs
the ciphertext with its private key.
Once the host receives all the shares, he combines them to retrieve the
$b$ and then uses an one way function (e.g., SHA3-512~\cite{bertoni2009keccak})
to convert it to a fixed length string.

\subsection{Signing}\label{sec:sign}
Apart from key generation, encryption and decryption, \name also 
supports distributed signing -- an operation that potentially manipulates 
a long term signature key. Here we introduce a novel multi-signature scheme, based on Schnorr signature~\cite{schnorr1991efficient}.

A multi-signature scheme allows a group of signers to distributively compute a signature on a common message.
There has been a very long line of works on multi-signature schemes \cite{OO91,MH96,MOR01,B03,CCS:BelNev06,LOSSW13} featuring several security properties (e.g. accountability, subgroup-signing) and achieving various efficiency trade-offs.  
A significant portion of work has been dedicated in reducing the trust in the key generation process. However, this  often involves the use of expensive primitives or increases the interaction between parties \cite{MOR01,B03,LOSSW13}. In our construction, we rely on \name's DKPG for the key generation process.

Our multi-signature scheme is based on Schnorr signatures \cite{schnorr1991efficient} and has similarities with Bellare and Neven \cite{CCS:BelNev06} construction.  One crucial difference between existing multi-signature schemes and ours, is that we utilize a host intermediating between signers (i.e., ICs). 
This intermediating host allows us to eliminate any interaction between the ICs and thus improve the efficiency of the construction. To further minimise the interaction between the host and ICs we adapt existing techniques used in proactive two-party signatures \cite{NKDM03} into the multi-signature context.

We let $PRF_s(j)$ be a pseudorandom function with key $s$, that takes $j$ as input and instantiates it as $Hash(s||j)$.

\begin{algorithm}
\SetAlgoLined
\SetKwInOut{Input}{Input}
\SetKwInOut{Output}{Output}
\Input{
The digest of the plaintext to be signed $H(m)$, the IC's private key of ${x_i}$ and secret $s$, an index $j$, and the aggregated random EC point ${R_j}$.}
\Output{The signature share tuple ${(\sigma_{ij},\epsilon_j)}$}
$\epsilon_{j} \gets Hash(R_j||Hash(m)||j)$\\
$r_{ij} \gets PRF_s(j)$\\
$\sigma_{ij} \gets r_{ij} -x_i \cdot \epsilon_j \mod n$\\
\Return ${(\sigma_{ij}, \epsilon_j)}$\\
\caption{SigShare: Returns the signature share of the IC for a given plaintext and index $j$.}
\label{alg:SigShare}
\end{algorithm}



Initially, all ${k}$ ICs cooperate to generate a public key ${y}$ using the distributed key generation operation (Section~\ref{sec:keygen}), and store securely their own key share ${x_i}$. Moreover, each IC generates a secret $s$ for the PRF, and stores it securely. After these steps, the signing protocol can be executed. The protocol comprises of two phases: \textit{caching} and \textit{signing}.

In the \textit{caching phase} (Figure~\ref{fig:SIGN_caching}), the host queries the ICs for random group elements $R_{ij}$, where $i$ is the id of the IC and $j$ an increasing request counter (\ding{182}).
Once such a request is received, the IC verifies that the host is authorized to submit such a request and then applies the keyed pseudorandom function on the index $j$ to compute $r_{i,j}=PRF_s(j)$ (\ding{183}). 
The IC then uses $r_{i,j}$ to generate a group element  (EC Point) $R_{ij}=r_{i,j}\cdot G$ (\ding{184}), which is then returned to the host.
Subsequently, the host uses Algorithm~\ref{alg:PointAgg} to compute the aggregate ($R_j$) of the group elements (Algorithm~\ref{alg:PointAgg}) received from the ICs for a particular $j$, and stores it for future use (\ding{185}).
It should be noted that the storage cost for $R_j$ is negligible: for each round the host stores only $65$ Bytes or $129$ Bytes  depending 
on the elliptic curve used (for $R_j$) and the corresponding round index $j$. This allows
the host to run the caching phase multiple times in parallel, and generate a list of random elements that can be later used, thus speeding up the signing process.

The \textit{signing phase} (Figure~\ref{fig:SIGN_v2}) starts with the host sending a Sign request to all ICs (\ding{182}). Such a request includes the hash of the plaintext $Hash(m)$, the index of the round $j$, and the random group element $R_j$ corresponding to the round. Each IC then first verifies that the host has the authorization to submit queries (\ding{183}) and that the specific $j$ has not been already used (\ding{184}). The latter check on $j$ is to prevent attacks that aim to either leak the private key or to allow the adversary to craft new signatures from existing ones. If these checks are successful, the IC executes Algorithm~\ref{alg:SigShare} and generates its signature share (\ding{185}). The signature share $(\sigma_{i,j},\epsilon_j)$ is then sent to the host (\ding{186}). Once the host has collected all the shares for the same $j$, can use Algorithm \ref{alg:PointAgg} on all the $\sigma_{i,j}$ to recover $\sigma_j$, obtaining the aggregate signature $(\sigma_j, \epsilon_j)$ (\ding{187}).

\begin{figure}
\centering
\includegraphics[width=\columnwidth]{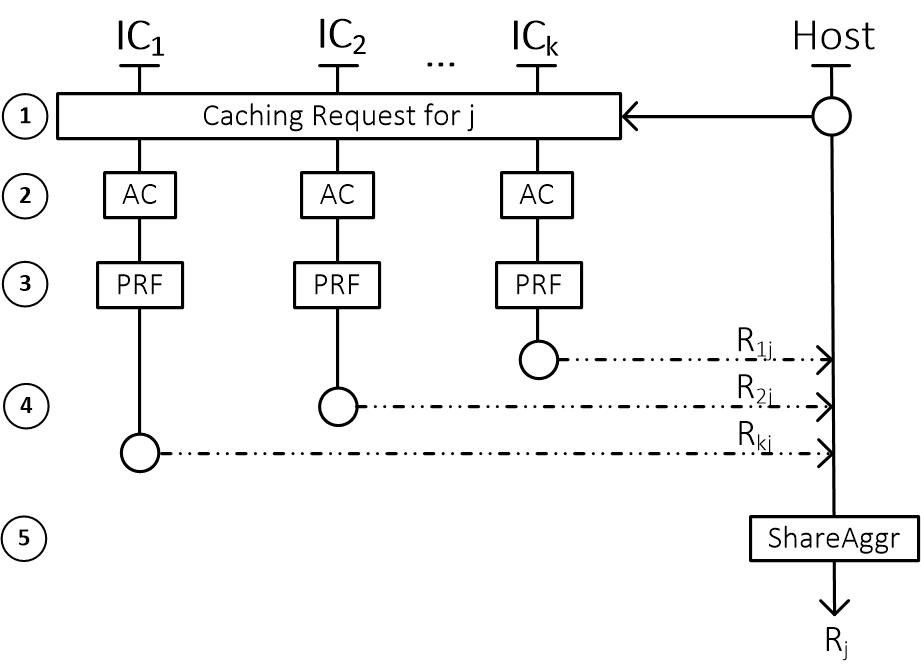}
\caption{The interaction between the different players during the caching phase of the distributed
  signing protocol.}
\label{fig:SIGN_caching}
\end{figure}


\begin{figure}
\centering
\includegraphics[width=\columnwidth]{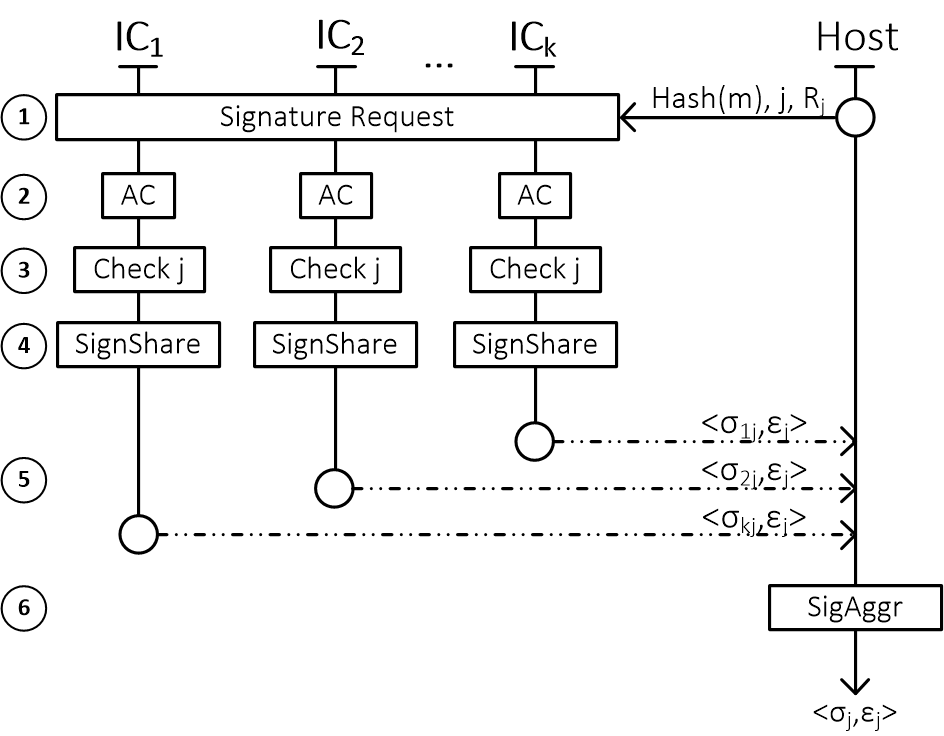}
\caption{The figure illustrates the interaction between the different players when signing.}
\label{fig:SIGN_v2}
\end{figure}

The recipient of ${\langle (m,j),\ \sigma, \epsilon \rangle}$ can verify the validity of the signature
by checking if $\epsilon=Hash(R||Hash(m)||j)$, where  $R = \sigma \cdot G + \epsilon \cdot Y$.


\medskip \noindent{\bf{Security Analysis.}}
The security of a multi-signature scheme requires that if at least one of the signers behaves honestly, then it is impossible to forge a multi-signature. In our context, the presence of a single honest IC guarantees the signature cannot be forged even in presence of a malicious host controlling all the remaining ICs. The key generation process is a crucial step for the security of multi-signature schemes due to the so called \emph{rogue-key attacks}~\cite{CCS:BelNev06}. The latter enables an attacker to maliciously generate shares of the public key in such a way that is possible for her to forge multi-signatures. In \name's DKPG process, malicious ICs cannot subvert the key generation process as long as at least one IC is not colluding with the others, thus preventing rogue-key attacks.
 The security of our multi-signature scheme is proved in Theorem~\ref{thm:sign} assuming the hardness of on the one-more discrete logarithm problem \cite{BNPS03}.  We refer to Appendix~\ref{sec:proof} for the security definitions of multi-signatures and the proof of the following Theorem.
 \begin{restatable}{theorem}{secproof}\label{thm:sign}
If there exists a $(q_S,q_H,\epsilon)$-forger $\F$ for the multi-signature scheme described in Figures~\ref{fig:SIGN_caching} and \ref{fig:SIGN_v2} interacting in $q_S=O(1)$ signature queries, making at most $q_H$ Hash queries and succeeding with probability $\epsilon$, then there exists an algorithm $\A$ that solves the $(q_S+1)$-DL problem with probability at least $$\delta \geq \frac{\epsilon^2}{q_{H}^{q_S+1}}-negl(\lambda)$$ 
\end{restatable}

\subsection{Key Propagation}\label{sec:keyprop}
In cases where more than one quorum is available it is useful to enable them all
to handle requests for the same public key. This is of particular importance for both the system's scalability and its availability, as we further discuss in Sections~\ref{sec:scalability} and ~\ref{sec:considerations} respectively.

Once a quorum $Q_1$ has generated its public key $y$,
(Section~\ref{sec:keygen}) the operator can specify another quorum $Q_2$ that is to be loaded with $y$.
Each member $q_i$ of $Q_1$ then splits its secret $x_i$ in $|Q_2|$ shares and distributes them to the
individual members of $Q_2$. To do that $q_i$ follows the secret sharing method shown in Algorithm~\ref{alg:simpless}.
However, any $t$-of-$t$ secret sharing schemes proposed in the
literature~\cite{blakley1979safeguarding,shamir1979share,DBLP:conf/crypto/Pedersen91} would do.

\begin{algorithm}
\SetAlgoLined
\SetKwInOut{Input}{Input}
\SetKwInOut{Output}{Output}
\Input{The domain parameters $T$, a secret $s$ which is to be shared, and the number of shares $k$.}
\Output{A vector of shares $\vec{v}_{s}$}
\For{($i=0$ to $k-2$)}{
$\vec{v}_s[i] \gets Rand(T)$
}
$\vec{v}_s[k-1] \gets (s-\vec{v}_s[1] - \vec{v}_s[2] - ... - \vec{v}_s[k-2])$\\
\Return $\vec{v}_s$\\
\caption{SecretShare: Returns a vector of shares from a secret.}
\label{alg:simpless}
\end{algorithm}

Once each member of $Q_2$ receives $|Q_1|$ shares, which they then combine to retrieve their share of the secret
corresponding to $y$. 
Each member of $Q_2$ can retrieve its share by summing the incoming shares, modulo $p$ (the prime provided in the domain parameters $T$).
An additional benefit of such a scheme is that $Q_1$ and $Q_2$ may have different sizes.

It should be also noted that a naive approach of having each member of $q_1$
send their share of $x$ to a member of  $q_2$ is insecure, as malicious members from
$q_1$ and $q_2$ can then collude to reconstruct the public key. 

\section{Implementation}\label{sec:impl}
In this section, we provide the implementation details of our \name prototype. We first focus on the custom 
hardware we built, and outline its different components and capabilities. Thereafter, we discuss the 
development of the software for the untrusted ICs and the details of the remote host.

\begin{figure}[h]
  \centering
  \includegraphics[width=\columnwidth]{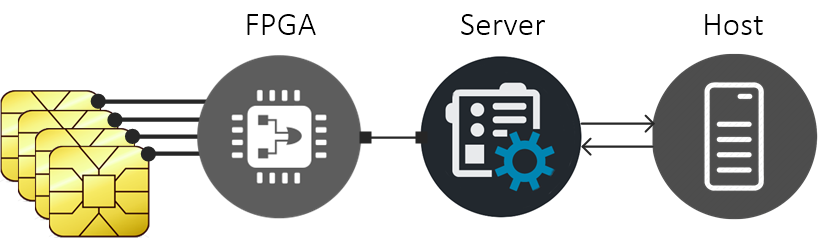}
  \caption{Overview of \name's components.}
  \label{fig:setup}
\end{figure}

\subsection{Hardware Design \& Implementation}

For our \name prototype, we designed our own custom circuit board, which features 120 processing ICs (set to use 40 quorums of 3 smartcards from different manufacturers) interconnected with dedicated buses with 1.2Mbps bandwidth.

The processing ICs are JavaCards (version 3.0.1), loaded with a custom applet implementing the protocols outlined in Section~\ref{sec:protocols}. JavaCards are an suitable platform as they provide good interoperability (i.e., applets are manufacturer-independent), which contributes to IC-diversification and prevents lock-in to particular vendors. Moreover, they also fulfill all the requirements listed in Section~\ref{sec:sec_mod}: (1) they are tamper-resistant (FIPS140-2 Level 4, CC EAL4) and can withstand attacks from adversaries with physical access to them~\cite{rankl2004smart}, (2) they feature secure (FIPS140-2 compliant) on-card random number generators, (3) they offer cryptographic capabilities (e.g., Elliptic curve addition, multiplication) through a dedicated co-processor and (4) there are numerous independent fabrication facilities (Section~\ref{sec:eval}). In addition to these, they have secure and persistent on-card storage space, ideal for storing key shares and protocol parameters.



The host is implemented using a computer that runs a python client application, which submits the user requests (e.g., Ciphertext Decryption) to \name using a RESTful API exposed by the device.
The incoming requests are handled by a Spring server, which parses them,
converts them to a sequence of low-level instructions, and then forwards these to the IC controller, through an 1Gbps TCP/UDP interface.
The ICs controller is a programmable Artix-7 FPGA that listens for incoming instructions
and then routes them to the processing IC, through separate physical connections.
We took special care that these buses offer a high bandwidth (over 400kbps), to prevent bottlenecks between controller and ICs even under extreme load. Once the ICs return the results, the controller communicates them back to the server, that subsequently forwards them to the host.

\begin{figure}
\centering
\includegraphics[width=\columnwidth]{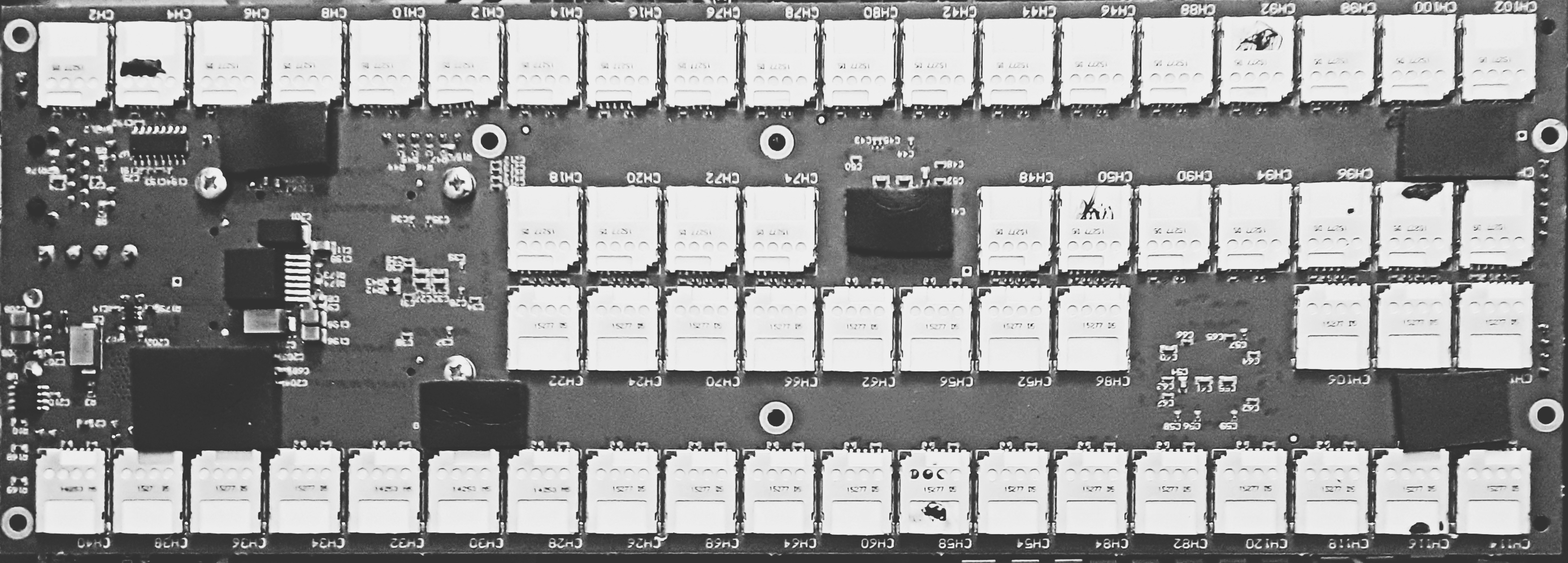}
\caption{\name's smartcard board supports 120 ICs (60 on each side). Our configuration splits them in 40 quorums of 3 diverse ICs each.}
\label{fig:board}
\end{figure}



\subsection{Software}\label{sec:soft}
We implement the protocols of Section~\ref{sec:protocols} and provide the necessary methods for inter-component communication and system parameterization.

We develop and load the processing ICs with JavaCard applets implementing methods for 1) Card Management, 2) Key Generation, 3) Decryption, and 4) Signing. Although JavaCard APIs since version 2.2.2 specifies a BigNumber class, this class is either unsupported by real cards or provides only too small internal type length (typically at most 8 bytes).  The only third-party BigInteger library available (i.e., BigNat\footnote{\url{https://ovchip.cs.ru.nl/OV-chip_2.0}}) is unmaintained and lacked essential operations. 
Moreover, low-level elliptic curve operations are only poorly supported by standard JavaCard API. The IC vendors often offer own proprietary APIs which provides the required operations (e.g., EC point addition) -- but for the price of reduced portability.  
This made the implementation of our cryptographic schemes more complicated. 

We extend BigNat to provide methods catering to our specific needs. Additionally, we develop own EC operations library based solely on public JavaCard API to support ICs where proprietary API is not available or cannot be used. Our EC library provides methods for EC point initialization, negation, addition, subtraction and scalar multiplication, and has been released as an independent project~\footnote{https://OpenCryptoJC.org}.

Note, that although vendor's proprietary API limits the portability, it usually provides better performance and better protection against various side-channel attacks in comparison to custom software-only implementations. For this reason, we structured the IC applet code for easy incorporation of proprietary API with minimal code changes. Our EC library is thus used only when no other option is available for target IC.    
 
Our current implementation is based on the NIST P-256~\cite{turner2009elliptic, adalierefficient} curve that provides at least 128 bits of security. However, it can also be trivially adapted for any other curve.


\medskip \noindent{\bf{Optimizations.}}
We optimize our JavaCard applet for speed and to limit side-channel attacks. Although JavaCard applets are compiled with standard Java compiler, common Java implementation patterns (e.g., frequent array reallocation due to resizing) are prohibitively expensive on JavaCards. Therefore, we use the following optimization techniques based on good practices and performance measurements from real, non-simulated, smart cards~\cite{vsvendanuances}: 
\begin{itemize}
\itemsep-0.025em
\item We use hardware accelerated cryptographic methods from JavaCard API instead of custom implementations interpreted by JavaCard virtual machine where possible.
\item We store the session data in the faster RAM-allocated arrays instead of persistent, but slower EEPROM/Flash
\item We use copy-free methods which minimize the move of data in memory, and also utilize native methods provided by JCSystem class for array manipulation like memory set, copy and erase.
\item We made sure to pre-allocate all cryptographic engines and key objects during the applet installation. No further allocation during the rest of the applet lifetime is performed.
\item Surplus cryptographic engines and key objects are used to minimize the number of key scheduling and initialization of engine with a particular key as these operations impose non-trivial overhead~\cite{vsvendanuances}.  
\item We refrain from using single long-running commands which would cause other cards to wait for completion, thus increasing the latency of the final operation. 
\end{itemize}

Finally, we optimized two fundamental operations commonly used in cryptographic
protocols: 1) integer multiplication, and 2) the modulo operation optimized for 32 byte-long EC 
coordinates. This was necessary, as the straightforward implementation
of these two algorithms in JavaCard bytecode is both slow
and potentially vulnerable to side-channel timing attacks. Instead, we
implemented both operations so as to use the native RSA engine and 
thus have constant-time run-time.

\sloppy
The integer multiplication of ${a}$
and ${b}$ can be rewritten as ${a\cdot b = ((a+b)^2 - a^2 - b^2) / 2}$. The squaring operation (e.g., ${a^2}$) can be quickly computed using
a pre-prepared raw RSA engine with a public exponent equal to 2 and
a modulus ${n}$, that is larger than the sum of the lengths of both operands.
On the other hand, the integer modulo of 
two operands ${a}$ (64 bytes long) and ${b}$ (32 bytes long)
is not so straighforward. We exploit the fact that ${b}$ is always the 
order of the fixed point G in the P-256 Elliptic Curve~\cite{turner2009elliptic, adalierefficient}, and
transform ${a \mod b = a - (((a \cdot x) \gg z) \cdot x)}$ where ${x}$ and ${z}$ values are pre-computed offline~\cite{granlund1994division}. As a result, a modulo operation
can be transformed to two RSA operations, one shift 
(with $z$ being multiple of 8) and one subtraction.
Note that we cannot directly use RSA with a public exponent equal to 1
as operands are of different length and also shorter than smallest RSA
length available on the card.

\subsection{System States} \label{sec:states}
Initially, the system is in an non-operational state, where the processing ICs do not 
respond to user requests. To make it operational, a secure initialization process
has to be carried out. During the initialization 
the processing ICs and the other components
described in~\ref{sec:arch} are loaded with
verified application packages, and the domain parameters
for the distributed protocols are set.
Moreover, the ICs are provided with their certificates that they will
later use to sign their packets and establish secure communication channels.

Once the initialization process has been completed, the system switches to an operational state
and is available to serve user requests. Depending on the configuration,
the system may be brought back to an uninitialized state, in order to load new software or change the
protocol parameters. We further discuss the importance 
of the initialization process in Section~\ref{sec:considerations}.


\section{Evaluation}\label{sec:eval}
In this section, we evaluate \name by examining both its performance,
and its qualitative properties.

\medskip \noindent{\bf{Experimental Setup.}}
All evaluations were performed using the setup illustrated in Figure~\ref{fig:setup}.
The host is a single machine with a CentOS 7 OS (3.10.0 Linux kernel),
featuring a 3.00GHz Intel(R) Core i5-4590S CPU, 16GB of RAM, and uses a python client application to
submit service requests to \name, through a 1Gbps Ethernet interface (as described in Section~\ref{sec:impl}).
Upon reception, the server uses the Java Spring framework~\cite{johnson2005introduction} to parse
them, and then forward the instructions to the Artix-7 FPGA, which subsequently routes them to the individual smart cards.
In all experiments, we collect response-time measurements from both the host and the Spring server. On average the round-trip from the host to the server takes ${5ms}$.
For accuracy, we use the host measurements in the rest of the section.

\subsection{Performance Impact}\label{sec:overhead}
This subsection evaluates the performance impact of \name, and compares
its throughput and latency with that of a single-IC system. Moreover, it
examines the impact of our optimizations on the overall performance of the 
system.

\medskip \noindent{\bf{Methodology.}}
To better understand the overhead that the use of a distributed architecture
entails we run experiments that measure the latency as the size of the protocol quorum grows.
We first submit 1,000 requests for each cryptosystem operation (Section~\ref{sec:protocols})
in one JavaCard and measure the response time.
We then extend the experiment to larger quorums with sizes ranging from 2 to 10, and measure the latency in completing the same 1,000 operations. 
Simultaneously, to gain a more in-depth understanding of the 
impact that each low-level instruction type has,
we micro-benchmark the response time for all intra-system communications.

\begin{figure}
  \centering
  \includegraphics[width=\columnwidth]{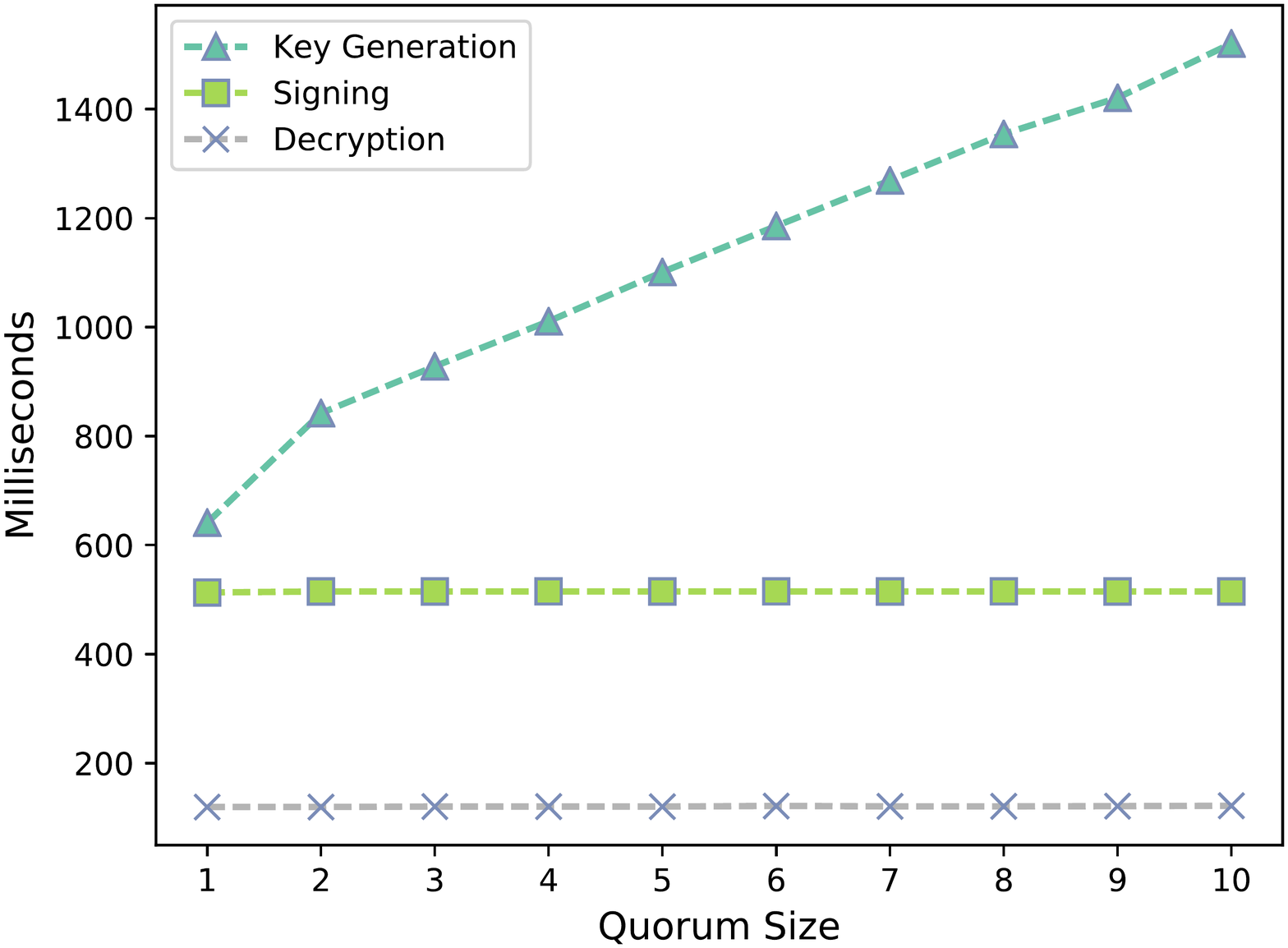}
  \caption{The average response time for each distributed operation of the cryptosystem, in relation to the quorum (i.e., a coalition of multiple ICs) size.}
  \label{fig:growth}
\end{figure}

\medskip \noindent{\bf{Results.}}
Figure~\ref{fig:growth} plots the average response time for performing Key Generation, Decryption and Signing using IC quorums of different sizes. 
To begin with, Decryption is the fastest ($119$ms), since it implements a single round protocol. Moreover, when we compare the runtime between the single-IC run, and the runs with multiple ICs, we observe that the latency is stable and the overhead remains always smaller than ${0.8}$\%. Hence, we conclude that the Decryption time does not increase with the size of the quorum, due to the ICs performing the decryption computations simultaneously. This highlights that Decryption is only CPU bound, and the network capacity of our prototype does not pose a bottleneck; and demonstrates that \name can essentially provide increased assurance, with negligible impact on the decryption runtime. It should be noted that, high-throughput decryption is of extreme importance in applications such as secure key derivation in mix-network infrastructures~\cite{danezis2009systems} that heavily rely on decryption. Similarly, the runtime for signing remains the constant ($\sim517$ms) regardless of the quorum size. This is because our multi-sig signing protocol does not require one-to-one interaction between the ICs.
The runtime difference between decryption and signing is mainly because of JavaCard API limitations. Specifically, we were forced to perform some of the mathematical operations for signing in the general purpose processor of the card, and not in the cryptographic coprocessor. The caching phase of the signing protocol (takes on average $169$ms/operation) was performed offline for thousands of indices $j$ and is not considered in the measurements, as it is ran only occasionally.  

As with signing and decryption, the runtime for random number generation is also constant, as the protocol is single round and can be executed in parallel.

\begin{figure}
  \centering
  \includegraphics[width=\columnwidth]{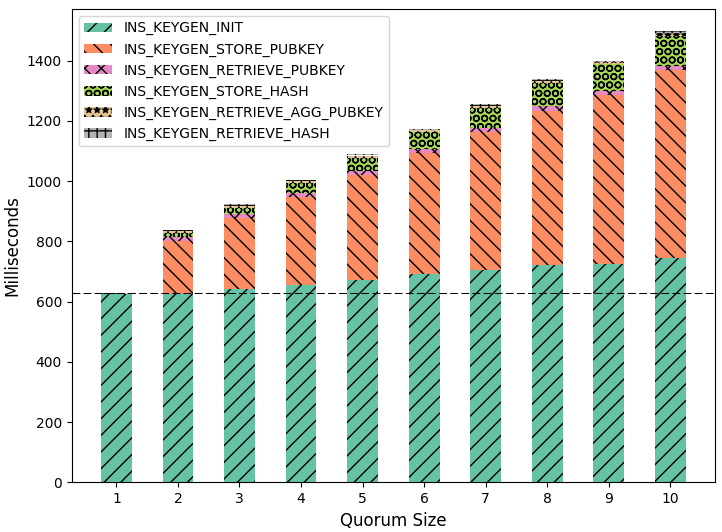}
  \caption{Breakdown of the runtime for low-level instructions that comprise the key generation operation, in relation to the quorum size. The horizontal reference line represent the cost of using a single IC.}
  \label{fig:keygen_exp}
\end{figure}

On the other hand, Key Generation (DKPG) requires two phases (commitment and revelation) and this adds significant latency.  In particular, as seen in Figure~\ref{fig:growth}, each additional IC results in an runtime increase of approximately 90ms. 
Figure~\ref{fig:keygen_exp} examines the timings of low-level operations involved in the key generation protocols. When quorums are small the cost of key generation is dominated by the ``INS\_KEYGEN\_INIT'' operation that generates a secret share, and derives a public key (624ms). However, as the IC quorums grow the operations related to exchanging public keys (``INS\_KEYGEN\_STORE\_PUBKEY'') and commitments (``INS\_KEYGEN\_STORE\_HASH'') become significant, and for a quorum of 10 ICs, nearly doubles the cost of key generation. However, for quorums of 3 the impact on runtime is merely 303ms, compared to a single IC. Other low-level operations used in DKPG have negligible cost, regardless of the quorum size.

\subsection{Scalability \& Extensibility}\label{sec:scalability}
This section examines how the throughput of our prototype grows
as the more processing power (i.e., quorums) is added. The absence of bottlenecks in the
system is important to be able to provide high-availability in production
environments.


\medskip \noindent{\bf{Methodology.}}
To determine how efficiently our design scales in practice, we 
run a series of experiments that measure \name's throughput,
for varying numbers of processing quorums. As described in 
Section~\ref{sec:impl}, our board supports up to 120 processing
ICs which can be divided into multiple quorums and serve requests
simultaneously. To benchmark the scalability of the system,
on each iteration of the experiment we submit 10,000 requests for each high-level
operation supported by our cryptosystem, and measure its throughput.
However, this time we fix the quorum size ${k}$ to 3, and 
on each iteration we increase the number of quorums 
serving the submitted requests by one, until we involve 40 quorums,
which is the maximum number of 
$3$-IC quorums that our prototype can support.
For simplicity, we assign each processing IC to only one quorum. However,
it is also technically possible for an IC to participate in more than one
quorums, and store more than one secret shares.

\medskip \noindent{\bf{Results.}}
Figure~\ref{fig:parallel} illustrates the throughput of the \name system (in operations per second) as more of the ICs are used for processing transactions.
The maximum throughput of \name was $315$ops/sec for
Decryption and $77$ops/sec for Signing, when using all $40$ quorums.
We also observe that as the number of quorums increases, the performance increases linearly, at a rate of ${\sim9}$ requests per second per additional quorum for the Decryption operation, and ${\sim1.9}$ requests per second for Signing. This illustrates that the system is CPU bound, and the communication fabric between ICs and the host is not a bottleneck. Consequently, a system with our proposed
architecture can be easily tailored for different use cases to provide the throughput needed in 
each of them. It should be also noted that using a lower threshold $k < t$ does not affect the performance of the system. However, this may result in some ICs being idle longer. For this purpose, it would be beneficial if ICs participated in more than one quorum, thus minimizing their idle time.

\begin{figure}
  \centering
  \includegraphics[width=\columnwidth]{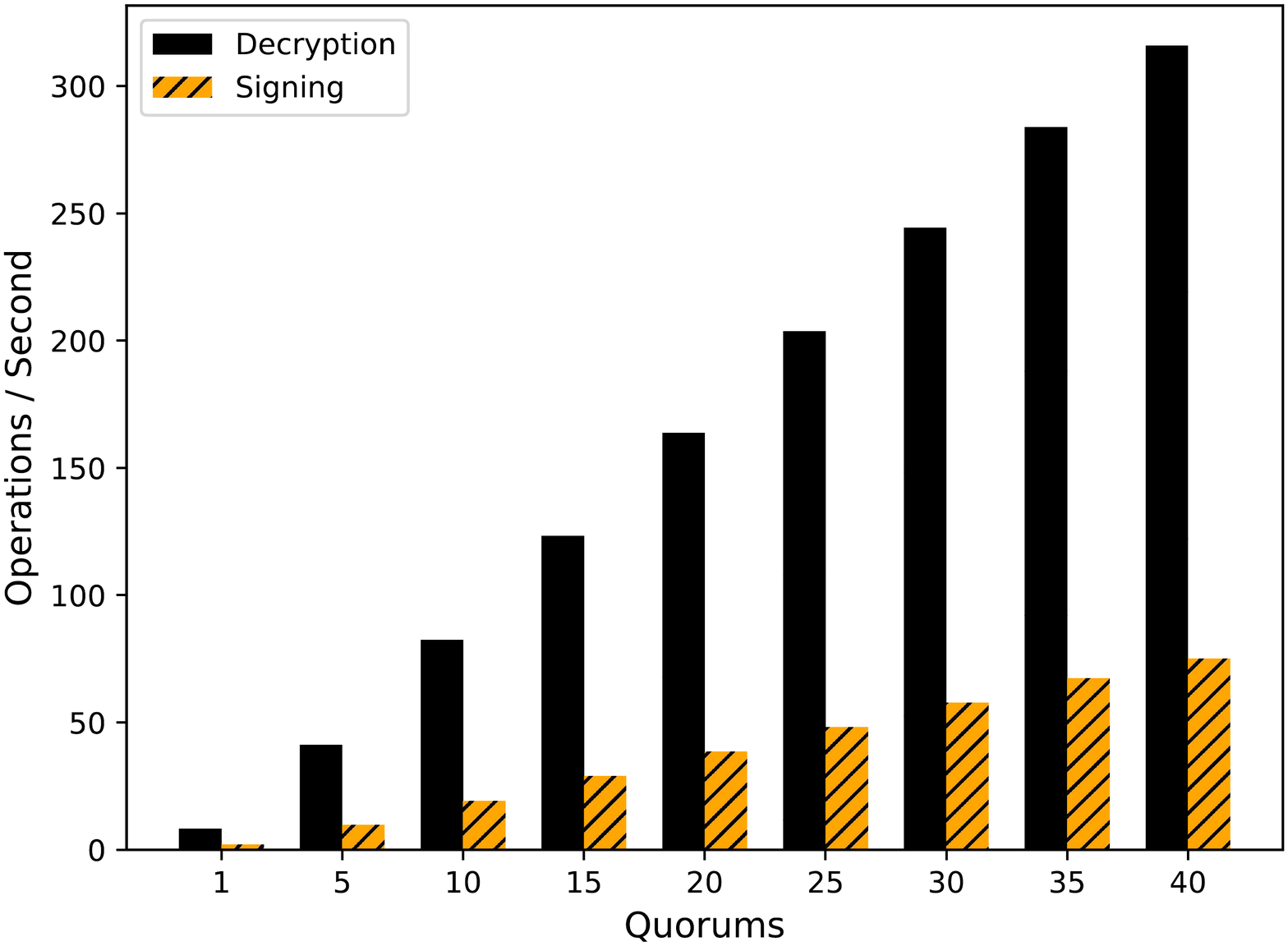}
  \caption{The average system throughput in relation to the number of quorums ($k=3$) that serve
  requests simultaneously. The higher is better.}
  \label{fig:parallel}
\end{figure}

\subsection{Tolerance levels}\label{sec:tradeoffs}
\name is tolerant against both fabrication-time and design-time attacks (assuming independent foundries and design houses). Alternatively, if there is only one design house, then \name protects against fabrication-time attacks
only. In this section, we examine the relationship between the system parameters and the tolerance levels provided depending on the attack type (Table~\ref{tbl:attacks}). The reported tolerance levels are those provided by the schemes outlined in Section~\ref{sec:protocols}, and were also empirically verified using our prototype introduced in Section~\ref{sec:impl}.

\begin{table}
\hyphenpenalty 10000

\begin{tabular}{c | c c c}
 \toprule
 \thead{Parameters} & \thead{Leakage} & \thead{Denial-of-Service} & \thead{IC Failures}\\
 \midrule
	$Single\ IC $ &   $0$ & $0$    & $0$ \\
	$k=t$ &   $t-1$ & $0$    & $n-1$ \\ 
	$k<t$ &   $k-1$ &  $t-k$	& $(t-k)*n$  \\ 
	\bottomrule \addlinespace
\end{tabular}
\caption{Number of malicious/faulty ICs that different setups can tolerate in each abnormal scenario. The system is assumed to feature $n$ identical quorums of size $t$, with a secret-sharing threshold $k$.}
\label{tbl:attacks}
\exhyphenpenalty 10000
\end{table}

In practice, the threshold ${k}$ and the size of the quorums $t$ express
the trade-off between confidentiality and robustness for each particular quorum. 
The relationship between these two parameters determines how many ICs 
can cease to function, before the quorum fails:
When ${k}$ equals the number of processing ICs ${t}$, 
then secrets are safe in the presence of ${t-1}$ compromised and colluding ICs.
On the other hand, a more ``relaxed'' threshold (${k<t}$) enables the quorum to remain fully functional unless more than $t-k$ ICs fail (maliciously or through wear).
Alternatively, ($k=t$)-systems can withstand multiple ICs failing (due to wear) using the technique introduced in Section~\ref{sec:keyprop}. This technique enables several quorums to handle requests for the same public key, and thus even if multiple ICs (and consequently the quorums they belong to) fail, the remaining quorums can still handle incoming requests for those public keys. 
It should be noted that this technique provides robustness only in the presence of faulty \& defective ICs, but does not mitigate denial of service attacks. This is because, all $n$ quorums are identical (same manufacturer diversity) and thus a malicious IC from a breached supply chain will participate in all quorums. In contrast, defective ICs will fail with a probability less than $1$ (post-fabrication tests detect always failing ICs) and thus several quorums will remain functional. 

From all the above, the security guarantees offered by \name are determined by the threshold ${k}$, the IC diversity and the number of quorums. In our prototype, we chose to maximize confidentiality, and resist the presence of ${t-1}$ actively malicious ICs. Malicious ICs launching denial-of-service attacks are easier to detect and replace, compared to those subtly leaking secrets or weakening keys. In cases where increased robustness and availability are paramount, the security level can be adjusted in favor of redundancy using the appropriate threshold schemes~\cite{DBLP:conf/crypto/Pedersen91}.

\subsection{Other Considerations}\label{sec:considerations}
In this section, we consider several qualitative properties commonly
used in the industry.

\medskip \noindent{\bf{Physical Security \& Diversity.}}
SmartCards forms the core of our prototype and have multiple benefits
as they were designed to operate in a hostile environment
that is assumed to be fully controlled by the adversary~\cite{rankl2004smart}.
For this reason, they come with very-high tamper-resistance (FIPS140-2, Level 4)
and secure storage capabilities, that are constantly evolving.
Another benefit is that there are several manufacturers owning foundries including NXP-semiconductors, ST Microelectronics, Samsung, Infineon and Athena. Moreover, there are numerous independent facilities used
by different vendors to fabricate their ICs~\cite{www-foundry1,www-foundry2,www-foundry3,www-foundry4}, which also contributes to the necessary diversity for the quorums.

\medskip \noindent{\bf{Code \& Parameter Provisioning.}} Crucial component of the security of \name. If the code on all the ICs, or the cryptographic parameters contain errors, then the security of the system is not guaranteed. We propose two strategies to ensure secure provisioning. First, we may assume that provisioning occurs at the factory. This leverages our assumption that some of the factories are honest, and ensures that an adversary would have to compromise all fabrication facilities to extract any secrets. The second strategy is to assume the existence of a trusted off-line provisioning facility that is not under the control of an adversary. This facility needs only to guarantee high-integrity, as no secrets are involved in the provisioning step (secrets are generated within the ICs as part of protocols executed later).

\medskip \noindent{\bf{Formal Security Validations.}} Any off-premise use of ICs raises a question of trust. To address this problem, independent security validations (e.g., FIPS140-2, or Common Criteria) have been introduced. These validations are performed by a third party, which verifies the claims of hardware vendors. However, these validations are a serious bottleneck when introducing new IC design. \name does not need to go through this process as it executes provably secure cryptographic schemes on already validated ICs. As a result, even if one of the ICs has passed an independent validation, the whole instance of the~\name carries this assurance. 

\medskip \noindent{\bf{Real-world Practicality.}}
\name provides a platform for generating legally binding digital signatures (eIDAS~\cite{www-eidas}) under the sole control of the user.  Moreover, due to the tamper-resistance properties of SmartCards, ICs can be also stored remotely thus making~\name a highly practical system that is able to support some non-typical real-world use-cases. For instance, \name can perform code-signing of mobile applications for apps stores (e.g., AppStore or Google Play), by sharing the signing keys between developers' laptops, managers' computers, and securely stored ICs providing protection against internal enterprise/company attackers without compromising developers' control over the signing process. Finally, another practical feature of \name is that each party maintains its own independent audit logs, thus ensuring non-repudiation.

\section{Related Work}\label{sec:related}
This section examines existing literature on hardware trojans 
and countermeasures, outlines relevant fault-tolerant designs and architectures, and
discusses how \name compares to prior works.


\medskip \noindent{\bf{Hardware Trojans \& Countermeasures.}}
To better understand the extend of the Hardware Trojans threat, numerous attacks and exploitation techniques have been proposed in the literature in the last decade.
For instance, the authors in \cite{kumar2014parametric} design two hardware trojans and 
implement a proof-of-concept attack against the PRINCE cipher~\cite{borghoff2012prince}.
The novelty of their attacks is that they use dopant polarity changes (first introduced in~\cite{becker2013stealthy}),
to create a hard-to-detect fault-injection attack backdoor. 
Moreover, \cite{pellegrini2010fault} also introduces a hardware trojan attacking RSA applications. 
In this attack, the adversary is able to use power supply fluctuations to trigger the trojan, which then leaks bits of the key through a signature.
Another very hard to detect class of trojans (inserted by fabrication-time attackers) was introduced by K Yang et al. in~\cite{yang2016a2}. Such trojans leverage analog circuits and require only a single logic gate to launch their attack (e.g., switch the CPU's mode). Apart from these, detection evasion and stealthy triggering 
techniques have been proposed in~\cite{DBLP:conf/dft/WangMKNB12,wang2014hardware,wang2011sequential,DBLP:conf/cases/KutznerPS13,DBLP:journals/pieee/BhuniaHBN14}. 

As discussed in Section~\ref{sec:intro}, malicious components carrying errors
have been also observed in commercial and military hardware~\cite{DBLP:conf/ches/SkorobogatovW12, mitra2015trojan, jin2010hardware, gallagher2014photos,skorobogatov2012hardware, adee2008hunt, markoff2009old, shrimpton2015provable}, while a subset 
of those incidents also involved misbehaving cryptographic hardware~\cite{www-freebsd, www-rt, www-ars, www-schneier}. 
In all these cases, the errors were eventually attributed to honest design or fabrication mistakes, but the systems were left vulnerable to attacks regardless.
For instance, one popular and well-studied example of attacks against weak cryptographic hardware is~\cite{bernstein2013factoring}.
In this work, Bernstein et al.\ study the random number generators used in smartcards and found various malfunctioning pieces, that allowed them to 
break 184 public keys using in ``Citizen Digital Certificates'' by Taiwanese citizens.

To address the aforementioned threats, different approaches have been proposed. The most common ones attempt to either detect malicious circuitry, or prevent insertions. In particular, detection techniques aim to determine whether any HTs exist in a given circuit and feature a wide range of methods such as side-channel analysis  \cite{DBLP:journals/tvlsi/WeiP14, 5778966, soll2014based, agrawal2007trojan},
logic testing~\cite{DBLP:conf/ches/ChakrabortyWPPB09},
and trust verification~\cite{rajendran2017logic, DBLP:conf/host/BaoXS15, DBLP:conf/ccs/WaksmanSS13, DBLP:journals/tcad/ZhangYWLX15}.
On the other hand, prevention techniques aim to either impede the introduction of HTs, or 
make HT easier to detect, such approaches are 
Split manufacturing~\cite{7835561, DBLP:conf/dac/WangCHR16, rajendran2013split} which tries to minimize the circuit/design exposure to the adversary,
logic obfuscation~\cite{chakraborty2009security}
and runtime monitoring~\cite{DBLP:conf/sp/HicksFKMS10, waksman2010tamper}

Moreover, there is also a smaller body of work which attempts to tackle the
even harder problem of inferring additional information about the HT such as its triggers, payload, and exact location~\cite{DBLP:journals/tvlsi/WeiP14,5778966}. Other works considered verifiable computation architecures (such as \cite{WHGSW16}), which provide guarantees for the correctness of the computation on untrusted platforms. However, they come with a computation overhead and do not guarantee secure handling of secrets or protection from side-channel attacks. On top of the above, previous works~\cite{dziembowski2016private, cryptoeprint:2016:527}  have also theoretically discussed using multi-party computation protocols to distribute trust between untrusted manufacturers during the fabrication process.

\name follows an alternative approach that leverages a diverse set of untrusted ICs to minimize the risk of compromises by distributing trust between them. For this reason, all the above countermeasures remain applicable and would even increase the security of the final system. In other words, our proposed approach is not an alternative to existing techniques, as it provides a way to combine ICs fabricated by different manufacturers, in various locations and featuring a wide-range of protection techniques in one system.

\noindent{\bf{Fault-Tolerant Systems.}}
Component-redundancy and component-diversification are both key concepts used in N-variant systems that 
aim to achieve high fault-tolerance~\cite{chen1978n}.
An example of such a system is the Triple-Triple Redundant 777 Primary Flight Computer \cite{yeh1996triple, yeh1998design}, that
replicates the computations in three processors and then performs a majority voting to determine the final result.
The applications of N-variance in security scenarios has been studied in only few works aiming to protect
systems against software attacks. In particular, \cite{DBLP:conf/uss/CoxE06} introduces a method to generate
randomized system copies, that will have disjoint exploitation sets, thus achieving memory safety.
On the other hand, \cite{DBLP:conf/dac/AlkabaniK08} proposes a N-variant method for IC diversification, 
aiming again to achieve disjoint exploitation sets. However, this method is not effective against fabrication-time attacks and protects only against (potentially exploitable) unintentional errors. Finally, heterogeneous architectures with COTS components have been 
also proposed in \cite{beaumont2012safer, beaumont2013hardware}. However, the
practicality of these works is very limited as the computations are simply replicated between the different components, thus not protecting against confidentiality attacks.

\section{Conclusion}\label{sec:concl}
High-assurance cryptographic hardware, such as hardware security modules,
is used in production under the assumption that its individual components are
secure and free from errors. However, even though post-fabrication testing 
is capable of detecting
defective ICs with high-accuracy, there are certain error classes that
are hard to uncover (e.g., RNG defects). Unfortunately, these errors
are also detrimental to the security of high-assurance hardware
such as hardware security modules. 
Moreover, hardware trojans and malicious circuitry 
have been extensively studied, and there is an abundance of 
designs, mitigation techniques and countermeasure evasion methods. 
This line of work assumes that not all errors can be detected and that due to the arms race,
between trojan horses and mitigation techniques, countermeasures will never be $100$\%
effective against all possible threats.

To resolve this, we introduce~\name, which brings the adversary into the 
unfavorable position of having to compromise 100\% of the hardware components
to gain any advantage.
By employing threshold schemes
and a redundancy-based architecture, \name is able to
distribute both secrets and cryptographic computations among multiple, diverse integrated circuits.
Consequently, an adversary aiming to breach the confidentiality or the integrity of the system,
must be able to compromise all the ICs. This is not a 
trivial task when the ICs 
are manufactured from different vendors, in different premises.
To evaluate \name, we build a custom board featuring
120 Smart Cards controlled by an Artix-7 FPGA.
The maximum throughput for distributed 
decryption is $315$ops/sec, while for signing is $77$ops/sec. Both come with
an overhead of less than $<1\%$ compared to a typical single-point of failure system.
All in all, our results show that \name is highly scalable, and provides strong
guarantees for the security of the system thus making it suitable for production.


\section{Acknowledgements}
This work was supported by the European Commission through the H2020-DS-2014-653497 PANORAMIX project and the European Research Council via the European Union's Seventh Framework Programme (FP/2007-2013) / ERC Grant Agreement n. 307937, and the Czech Science Foundation under project GA16-08565S.


\bibliographystyle{ACM-Reference-Format}
\bibliography{hwmpc}
\newpage
\appendix

\section{Multi-signature Scheme Proof}\label{sec:proof}
In this Appendix we discuss the security of our construction introduced in Section~\ref{sec:sign}. Before moving to the proof of the Theorem~\ref{thm:sign} we recall the definition of the one-more discrete logarithm problem \cite{BNPS03} and security definitions for multi-signatures schemes.

\begin{definition}[$N$-DL]
A group generator  $\GG(1^\lambda)$ is a probabilistic polynomial time algorithm that on input a security parameter $\lambda$ returns a pair $(p,G)$ where $p$ is a $\lambda$-bit prime and $G$ is a random generator of a cyclic group $\Gr$ of order $p$.

An algorithm $\A$ to solve the $N$-DL problem is a probabilistic polynomial time algorithm which receives as input an instance $(p,G)\gets \GG(1^\lambda)$ and can access two oracle $O_{ch}, O_{dlog}$. Upon request, the former returns a random group element in $\Gr$. The latter gets as input a group element $C$ and returns its discrete logarithm with respect to the generator $G$, i.e. such that $C=x\cdot G$. Let $C_1,C_2,\ldots, C_N$ be the challenges returned by $O_{ch}$. We say that adversary $\A$ wins if he returns $c_1,\ldots,c_N\in \Z_p$  satisfying $C_i=c_i\cdot G$ by using a number of queries to $O_{dlog}$ \emph{strictly} smaller than $N$. 
\end{definition}

Next, we recall definitions of multi-signatures of \cite{CCS:BelNev06}, adapted to our settings.

\begin{definition}[Multi-Signature]
A multi-signature scheme is a tuple $(KeyGen, Sign, Verify)$ of algorithms. 
\begin{itemize}
\item $KeyGen$: This is an interactive protocol between $n$ parties $IC_i$ to jointly compute a common shared verification key $Y$ and $n$ individual signing keys $x_i$ associated with $Y$.
\item $Sign$: This is an interactive protocol between the $n$ ICs that given input a common message $m$ and their individual secret keys $x_i$ allows to compute a shared signature $\Sigma$ on $m$.
\item $Verify$: This is a deterministic algorithm that given in input the common verification key $Y$, a signature $\Sigma$ and a message $m$ returns 1 if the signature is valid and 0 otherwise. 
\end{itemize}
\end{definition}

In the construction presented in Section~\ref{sec:sign} the key generation protocol is handled by the DKPG; $Sign$ is described by Algorithm~\ref{alg:SigShare} and Figures~\ref{fig:SIGN_caching} and \ref{fig:SIGN_v2}; $Verify$ consists of the verification algorithm of standard Schnorr signatures on the aggregated signature using the shared verification key $Y$.

\begin{definition}[Security Game]
We consider and adversary $\F$ attempting to forge a multi-signature. The attack is modelled as a game in three phases.
\begin{itemize}
\item Setup: We assume the key generation protocol among $n$ parties to be successfully executed and returning a shared public key $Y$ and a set of $n$ secret keys $x_i$.  
\item Attack: We assume the forger $\F$ to corrupt $n-1$ ICs and learn their own individual secret keys $x_i$. Without loss of generality, we assume user $IC_1$ to be the only honest user in the system. The forger $\F$ interacts as the host over $q_S$  interactive signing sessions with $IC_1$ and arbitrarily deciding on the messages to be signed. Let $Q$ to be the set of messages $m$ used in the interactive signing sessions.
\item Forgery: At the end of its execution, the forger $\F$ returns a signature $\Sigma$ on a message $m$ which was not used on a signing session with $IC_1$, namely $m\notin Q$. The forger $\F$ wins the game if $Verify(Y,m,\Sigma)=1$.
\end{itemize}

We model the security of the scheme in the random oracle model~\cite{fiat1986prove,BR93}. This means that we assume the hash function to behave as an ideal random function. This is modelled by giving to the adversary $\F$ access to an oracle $O_{Hash}$ returning random values in a range $\{0,1\}^\lambda$.

The advantage $adv_{ms}(\F)$ of algorithm $\F$ in forging a multi-signature is defined to be the probability that $\F$ wins the above game, where the probability is taken over the random coins of $\F$, $IC_1$, DKGP and the random oracle.
We say that $\F$ $(q_S, q_H,\epsilon)$-breaks the multi-signature scheme if it participates in $q_S$ signing sessions with $IC_1$, makes at most $q_H$ queries to $O_{Hash}$ and the advantage $adv_{ms}(\F)\geq \epsilon$.
\end{definition}

%
%
%
%
%
%

We now restate and prove Theorem~\ref{thm:sign}.
 \secproof*
%
\begin{proof}
We start by describing the idea behind the reduction. Assume for a moment that a forger $\F$ is able to produce two valid signatures $(\sigma,\epsilon)$ and $(\sigma',\epsilon')$ on the same message, i.e.
\begin{equation}\label{eq:condition1}\epsilon=Hash(\sigma\cdot G +\epsilon \cdot Y||H(m)||j) \qquad \epsilon'=Hash(\sigma'\cdot G +\epsilon' \cdot Y||m||j)\end{equation}
 and such that
\begin{equation}\label{eq:condition2}R=\sigma\cdot G +\epsilon\qquad R=\sigma'\cdot G +\epsilon' \cdot Y\end{equation}
Dividing the two equations we obtain $$(\sigma-\sigma')\cdot G=(\epsilon'-\epsilon)\cdot Y$$
then we get the discrete logarithm of $Y$ with respect to $G$ by computing $(\sigma-\sigma')*(\epsilon'-\epsilon)^{-1}\mod p$.

Given a  forger $\F$ to output forged signatures, we construct an adversary $\A$ for solving the $(q_S+1)$-$DL$ problem. In the process, $\A$ has to simulate signatures as produced by $IC_1$ and answer to the random oracle queries made by $\F$ during the attack, as it happen in the security game for multi-signatures.
In case $\F$ succeeds in forging a first signature, then adversary $\A$ rewinds $\F$ and replay him reusing the same coin tosses used in the first execution. However in the second run of $\F$, the adversary replaces the answer of the random oracle query corresponding to the forgery produced in the first execution with a fresh random string in the range of the hash function. By applying a forking lemma type of argument \cite{PS00} one can then argue that with good probability the replay of $\F$ will return a new forgery which has the same $R$ as in the first forged signature but different $(\sigma',\epsilon')$.
Once adversary $\A$ obtains two signatures of this kind, he will be able to compute the discrete logarithm of $Y$ as shown above\footnote{The adversary described in the proof actually has to computes the discrete logarithm of $Y_1$, the public key of $IC_1$. This is trivial to do once the adversary $\A$ has obtained the discrete logarithm of $Y$, the shared verification key.}.

The description of $\A$ is given in Algorithm~\ref{alg:adversary}. Throughout the execution, adversary $\A$ keeps track of the random oracle queries made by $\F$ using a list $L[\cdot]$. Without loss of generality we assume the adversary $\F$ checks all the signature queries made to her oracle, as well as the signature he attempts to forge. 

We can summarise the interaction of adversary $\A$ with $\F$ in the following phases. 

\begin{algorithm}
\SetAlgoLined

\SetKwInOut{Input}{Input}
\SetKwInOut{Output}{Output}
\Input{$(p,G)$.}
\Output{$q_S+1$ discrete logarithms using $q_S$ calls to $O_{dlog}$}
\BlankLine
\SetKwProg{setup}{Setup Phase}{}{}
\setup{}{
$ctr=0, j=1$\\
Set $L[\cdot]=\emptyset, J[\cdot]=\emptyset$ \\
Set random string $\rho$ for adversary $\F$.\\
Pick $\pi_1,\pi_2,\ldots, \pi_{q_H}\gets \{0,1\}^\lambda $\\
$Y_1 \gets O_{ch}$\\
\For{$i =2$ to $n$}{
        $x_i\gets \Z_p$\\
	$Y_i={x_i}\cdot G \gets O_{ch}$\\
	}
$Y=\sum_{i=1}^n Y_i$\\
}
\BlankLine
\BlankLine
\SetKwProg{start}{Start}{}{}
\SetKwProg{run}{Run}{}{}
\SetKwFunction{adv}{$\F$}
\start{\adv{$Y,(Y_1,\ldots, Y_n),(x_2,\ldots, x_n);\rho$}}{}

\BlankLine
\BlankLine
\SetKwProg{cache}{Caching Phase:}{}{}
\cache{$Init \gets \F$}{
$R_{1,1},R_{1,2},\ldots, R_{1,q_S}\gets O_{ch}$\\
\run{\adv{$R_{1,1},R_{1,2},\ldots, R_{1,q_S}$}}{} 
}
\BlankLine
\BlankLine
\SetKwProg{rom}{Random Oracle \text{Queries}:}{}{}
\rom{$(Hash,M)\gets \F$}{
\If{$L[M]=\bot$}{$ctr +=1$\\
	$L[M]=ctr$\\}
$k=L[M]$\\
\run{\adv{$\pi_k$}}{} 
}

%
\BlankLine
\BlankLine
\SetKwProg{sign}{Signing \text{Queries}:}{}{}
\sign{$(Sign,(R,H(m),j'))\gets \F$}{
\If{$ j'\neq j$}{\Return{$\bot$ to $\F$}
}
\Else{$j+=1$}
\If{$L[R||H(m)||j']=\bot$}
	{$ctr +=1$\\
	$L[R||H(m)||j']=ctr$
	}
$k=L[[R||H(m)||j']]$\\
$J[j']=k$\\
 $\epsilon_{j'}=\pi_{k}$\\
$\sigma_{j'}\gets O_{dlog}(R_{1,j'}-h\cdot Y)$\\ 
\run{\adv{$(\sigma_{j'},\epsilon_{j'})$}}{} 
}

\BlankLine

$(Forgery, (m,(\sigma,\pi_i)))\gets \F$\\
\If{$Verify(Y,m,(\sigma,\pi_i))=0$}{\Return $\bot$}
\BlankLine
\BlankLine
\SetKwProg{replay}{Replay}{}{}
\replay{\adv}{} 
\BlankLine
\BlankLine

$(Forgery, (m,(\sigma',\pi_i')))\gets \F$\\
\If{$Verify(Y,m,(\sigma',\pi_i'))=0$}{\Return $\bot$}
$x_1=(\sigma-\sigma')*({\pi_i}'-\pi_i)^{-1} -\sum_{j=2}^n x_j\mod p$\\
\For{$j =1$ to $q_S$}{
	$k=J[j]$\\
        $r_{1,j}=s_j+x_1*\pi_k \mod p$\\
	}
\Return{$(x_1,r_{1,1},\ldots,r_{1,q_S})$}
\caption{Adversary $\A_\F^{O_{ch},O_{dlog}}$ against $(q_S+1)$-DL}\label{alg:adversary}
\end{algorithm}

\begin{itemize}
\item Setup Phase: $\A$ initialises an empty list $L$. She picks a random string $\rho$ for the randomness used by adversary $\F$ as well as $q_H$ random values $\pi_i$ from the range of the hash function, hereby set to be $\{0,1\}^\lambda$. The adversary then simulates the multi-signature key generation process: she generates pairs of public and secret keys for the corrupted cards $IC_2,\ldots, IC_n$, as honest cards would do, and generates the public key of $IC_1$ by querying his own challenge oracle $Y_1\gets O_{ch}$. Finally, she generates the shared public key $Y$ as in the key generation protocol (Algorithm~\ref{alg:PointAgg}). Then she starts adversary $\F$  on inputs the shared verification key $Y$, the list of verification key-shares and the $n-1$ signing key corresponding to the corrupted ICs. 
\item Caching Phase: the adversary $\F$, acting as the hosts, initiates the signature process by sending $Init$ to $IC_1$. In this phase the adversary emulates the caching phase of $IC_1$ by making $q_S$ queris to $O_{ch}$ and returning the list of group elements to $\F$.
\item Random Oracle Queries: whenever $\F$ makes a new query $(Hash, M)$ to its random oracle, adversary $\A$ answers it by picking the next unused random string in $\{\pi_1,\pi_2,\ldots, \pi_{q_H}\}$. She keeps track of $\F$ requests to answer consistently.
\item Signing Queries: in this phase the forger submits signature queries as the host in the system. The adversary $\A$ hashes the message, consistently with the list of random oracle queries and uses her oracle $O_{dlog}$ to compute the requested signature.
\item Replay: $\F$ eventually attempts to output a forged multi-signature $(\sigma,\epsilon)$. If the signatures verifies, $\A$ rewinds adversary $\F$ to the beginning and replay it on input the same random coins $\rho$ and same inputs and reusing the same response in the caching phase. $\A$ answers the random oracle and signing queries as in the first execution apart form the random oracle query which generated $\epsilon$. $\A$ answers this query by picking a random string $\pi_i'\in\{0,1\}^\lambda$. The forger attempts to produce a new signature $(\sigma',\epsilon')$. If conditions  \eqref{eq:condition1} and \eqref{eq:condition2} are met, then $\A$ computes the discrete logarithm of $Y$, and thus the discrete logarithm of $Y_1$. Given the latter and the responses of $O_{dlog}$ she can also compute the discrete logarithm of the $R_{1,j}$, the outputs of $O_{ch}$.
\end{itemize}

In order to succeed in her game, adversary $\A$ has to simulate the interaction between $\F$ and $IC_1$ as described in the security game of multi-signatures.  The above adversary perfectly simulates the multi-signature key generation phase, the caching phase of $IC_1$ and the random oracle queries of $\F$. Extra care needs to be taken for signature queries simulating the interaction of $IC_1$ and $\F$.

\medskip \noindent{\bf{Simulation of Signature Queries.}}
Note that differently from standard Schnorr signatures, we cannot simulate a signature by programming the random oracle. This is because the caching phase of $IC_1$ fixes the randomness used in the signatures before $\A$ knows the messages corresponding to them, preventing him to program the random oracle accordingly. Thus, in order to simulate $IC_1$'s signatures $\A$ takes advantage of her $O_{dlog}$ oracle. As $\A$ has to make fewer queries to $O_{dlog}$ than to $O_{ch}$, she is allowed one call to the discrete logarithm oracle for each signature in the first round. The problem arises when $\A$ replays the forger $\F$. In this execution we want $\F$ to behave exactly as in the first round up to the point he queries (a second time) the random oracle on the forged message. This means that in the caching phase $\A$ is forced to reuse the same random group elements as in the first execution. Since $\A$ cannot make new queries to the challenge oracle, she cannot make any more queries to the $O_{dlog}$ during the second execution of $\F$. Therefore, in the second execution, $\A$ has to simulate signatures (on potentially different messages) without the help of her oracle. Assume for a moment that the index of hash queries used for the creation of signatures is the same in both executions. In this case, the response of the random oracle on the corresponding hash queries will be the same, even if the message has changed during the two execution. Since the the output of the hash is the same in both runs, $\A$ is able to answer the signature queries by simply recycling the signature produced in the first execution.

\medskip \noindent{\bf{Computation of Discrete Logarithm.}} 
Under the condition that $\A$ simulates $\F$'s environment, the above adversary succeeds to extract the discrete logarithm of $Y_1$ as long as the index of the hash query associated to the forged signatures is the same in both runs. In this case conditions \eqref{eq:condition1} and \eqref{eq:condition2} are met with overwhelming probability and $\A$ can proceed to compute discrete logarithm as illustrated above. 

From the above observations we can summarise two conditions which make adversary $\A$ to succeeds in her game:
\begin{enumerate}
\item The index $i$ of the random oracle query corresponding to the forged signature is the same in both executions of $\F$. 
\item The set $J$ of indexes of random oracle queries corresponding to signature queries is the same in both execution of $\F$.
\end{enumerate}

By assuming that $\F$ forges signatures with non-negligible probability $\epsilon$, the following lemma shows the above conditions holds with non-negligible probability. This concludes the proof as it gives a bound $\delta$, the probability of success of $\A$. 
\end{proof}
%
%
The following lemma is an adaptation of 
the generalized forking lemma of \cite{CCS:BelNev06} and Lemma 5 of \cite{NKDM03}.   

\begin{lemma}
Let $\F$ be an adversary producing a multi-signature forgery with probability $\epsilon$. Assuming that $\F$ makes at most $q_{H}$ hash queries and interacts in $q_S=O(1)$ signatures, then by replaying $\F$, $\A $ succeeds with probability $\delta \geq \frac{\epsilon^2}{q_{H}^{q_S+1}}-negl(\lambda)$.
\end{lemma}

\begin{proof}
Let $i \in [q_{H}]$ be the index of the random oracle query associated to the forgery returned by $\F$ and let $J\subseteq [q_{H}],|J|\leq q_S$ be the set of indexes of oracle queries associated with the signatures queries made by $\F$. Call $i$ and $J$ the \emph{target} indexes. Note that $J$ and $\{i\}$ are disjoint as a valid forgery cannot be on a message queried on the signature oracle.

For any fixed $i,J$ let  $\epsilon_{i,J}$ be the probability that $\F$ produces a forgery on target indexes $i,J$. The probability of $\F$ forging a signature is given by $\epsilon=\sum_{i,J}\epsilon_{i,J}$.

Consider now all the random inputs given to $\F$. These includes the random coins $\rho$, his inputs, the response of the caching phase, the responses of the random oracle queries and the signature queries. We split these into: the $i$-th response to the random oracle query, $\pi_i$, and everything else, which we call $R$. Consider now a matrix $M_{i,J}$ with one column for every possible answer to $\pi_i$ and one row for any possible string $R$. The entries of the matrix are 1, in case $R$ and $\pi_i$ will make $\F$ produce a forgery for target indexes $i,J$, and 0 otherwise. For any $i,J$, let $\epsilon_{i,J,R}$ be the probability of the adversary $\F$ to produce a forgery given randomness $R$. This corresponds to the density of the row $R$ in the matrix $M_{i,J}$. Similarly, $\epsilon_{i,J}$ corresponds to the density of the matrix $M_{i,J}$, namely
$$\epsilon_{i,J}=\frac{1}{|R|}\sum_{R}\epsilon_{i,J,R} $$  

In a similar fashion we define $\delta_{i,J}$ to be the probability of producing two forgeries on the same target indexes $i,J$ while replaying the adversary on same randomness and replacing $\pi_i$ with a random $\pi_i'\gets \{0,1\}^\lambda$. Similarly, we also define $\delta_{i,J,R}$ for which we have
$$\delta_{i,J}=\frac{1}{R}\sum_{R}\delta_{i,J,R} $$
 
 We now relate probability $\delta$ of adversary $\A$ to $\F$. The probability $\delta_{i,J,R}$ corresponds to the event of sampling randomness $R$ for the adversary $\F$, hitting a first 1 in the row $R$ of $M_{i,J}$ and then probe another random column in the same row and hitting another 1. Since the two runs are independent, the probability of succeeding in the replay attack is $$\delta_{i,J,R}=\epsilon_{i,J,R}*\left(\epsilon_{i,J,R}-\frac{1}{2^\lambda}\right)$$
Replacing the above in the expression of $\delta_{i,J}$ and applying the Cauchy-Schwartz inequality give the following
\begin{align*}
\delta_{i,J}&=\frac{1}{|R|}\sum_{R}\delta_{i,J,R}\\
&=\frac{1}{|R|}\left(\sum_{R}\epsilon_{i,J,R}^2-\frac{\epsilon_{i,J,R}}{2^\lambda}\right)\\
&\geq \left(\frac{1}{|R|}\sum_{R}\epsilon_{i,J,R}\right)^2-\frac{\epsilon_{i,J}}{2^\lambda}\\
&=\epsilon_{i,J}^2 -\frac{\epsilon_{i,J}}{2^\lambda} 
\end{align*}
\vfill\eject Given the above we get the following bound on $\delta$
\begin{align*}
\delta&=\sum_{i,J}\delta_{i,J}\geq \sum_{i,J}\left(\epsilon_{i,J}^2 -\frac{\epsilon_{i,J}}{2^\lambda}\right) \\
&=\sum_{i,J}\epsilon_{i,J}^2 -\frac{\epsilon}{2^\lambda}\\
&\geq \frac{1}{q_H^{q_S+1}}\left(\sum_{i,J}\epsilon_{i,J}\right)^2-\frac{\epsilon}{2^\lambda}\\
&=\frac{\epsilon^2}{q_H^{q_S+1}}-\frac{\epsilon}{2^\lambda} 
\end{align*}
 where the last inequality is obtained by applying again the Cauchy-Schwartz inequality.
\end{proof}

\section{Prototype Extensions}\label{sec:prototype}
Our \name prototype can also be extended to use multiple ICs boards, if higher throughput is needed. Figure~\ref{fig:fullboard} depicts an instantiation with 240 ICs: two boards holding 120 ICs each (60 ICs on each side) and an Artix-7 FPGA to facilitate the inter-IC communication within each board.

\begin{figure}[h]
\centering
\includegraphics[width=\columnwidth]{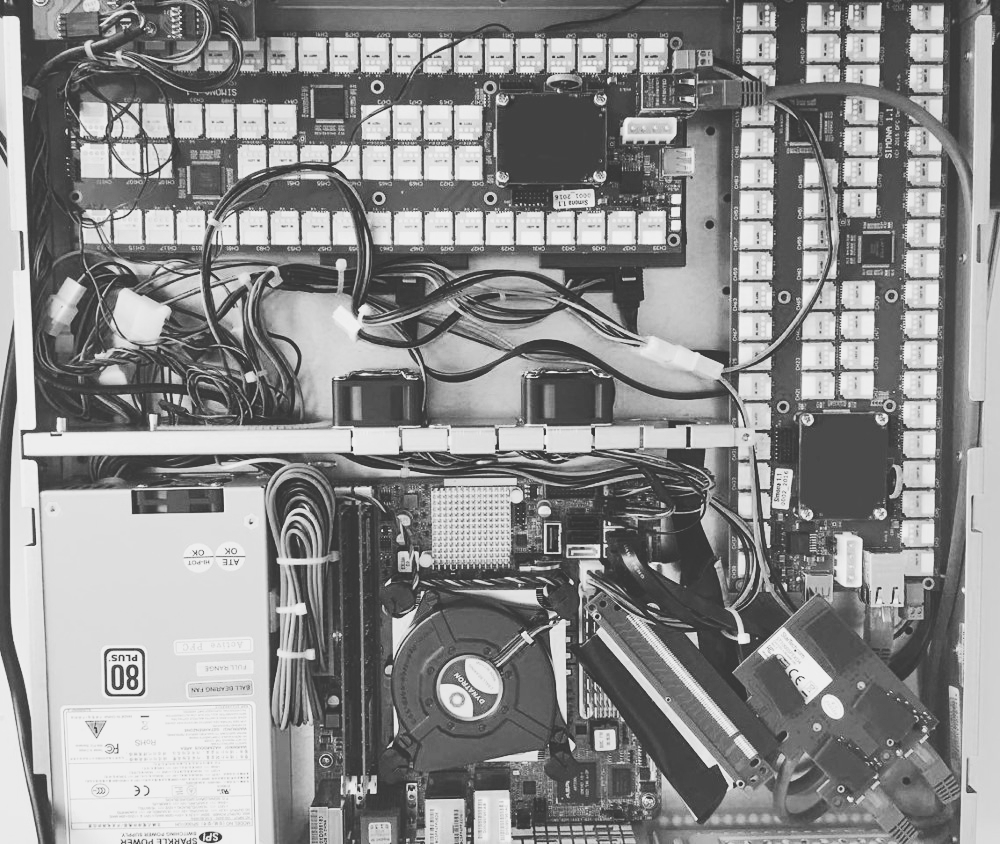}
\caption{\name's prototype with 240 JavaCards fitted into an 1U rack case.}
\label{fig:fullboard}
\end{figure}

\end{document}